\documentclass[12pt, draftclsnofoot, onecolumn]{IEEEtran}	

 \usepackage{amsmath}
 \usepackage{amssymb}
 \usepackage{subfigure}
 \usepackage{graphicx}
 \usepackage{graphics}
 \usepackage{color}
 \usepackage{psfrag}
 \usepackage{cite}
 \usepackage{balance}
 \usepackage{algorithm}
 \usepackage{accents}
 \usepackage{amsthm}
 \usepackage{bm}
 \usepackage{url}
 \usepackage{algorithmic}
 \usepackage[english]{babel}
 \usepackage{multirow}
 \usepackage{enumerate}
 \usepackage{cases}
 \usepackage{stfloats}
 \usepackage{dsfont}
 \usepackage{soul}
 \usepackage{amsfonts}
 \usepackage{fancyhdr}
 \usepackage{hhline}
 \usepackage{array}
 \usepackage{booktabs}
 \usepackage{framed}
 \usepackage{float}

\newtheorem{theorem}{Theorem}

\newtheorem{definition}{Definition}

\newtheorem{lemma}{Lemma}

\newtheorem{proposition}{Proposition}

\newtheorem{assumption}{Assumption}

\newtheorem{remark}{\bf Remark}
\def\phi{\varphi}

\def\l{\left}
\def\r{\right}
\def\({\left(}
\def\){\right)}

\def\b0{{\mathbf{0}}}

\newcommand{\nn}{\nonumber}


\newcounter{parentnumber}

\begin{document}

\title{\Large Energy-Efficient Resource Management for Federated Edge Learning with CPU-GPU Heterogeneous Computing}


\author{
Qunsong Zeng, Yuqing Du, Kaibin Huang, and Kin K. Leung 
\thanks{Q. Zeng, Y. Du, and K. Huang are with  The University of Hong Kong, Hong Kong. K. K. Leung is with Imperial College London, UK. Contact: K. Huang (huangkb@eee.hku.hk).}
}

\maketitle
\begin{abstract}
Edge machine learning involves the deployment of learning algorithms at the network edge to leverage massive distributed data and computation resources to train \emph{artificial intelligence} (AI) models. Among others, the framework of \emph{federated edge learning} (FEEL) is popular for its data-privacy preservation. FEEL coordinates global model training at an edge server and local model training at edge devices that are connected by wireless links. This work contributes to the energy-efficient implementation of FEEL in wireless networks by designing joint \emph{computation-and-communication resource management} ($\text{C}^2$RM). The design targets the state-of-the-art heterogeneous mobile architecture where parallel computing using both a CPU and a GPU, called \emph{heterogeneous computing}, can significantly improve both the performance and energy efficiency. To minimize the sum energy consumption of devices, we propose a novel $\text{C}^2$RM framework featuring multi-dimensional control including bandwidth allocation, CPU-GPU workload partitioning and speed scaling at each device, and $\text{C}^2$ time division for each link. The key component of the framework is a set of equilibriums in energy rates with respect to different control variables that are proved to exist among devices or between processing units at each device. The results are applied to designing efficient algorithms for computing the optimal $\text{C}^2$RM policies faster than the standard optimization tools. Based on the equilibriums, we further design energy-efficient schemes for device scheduling and greedy spectrum sharing that scavenges ``spectrum holes" resulting from heterogeneous $\text{C}^2$ time divisions among devices. Using a real dataset, experiments are conducted to demonstrate the effectiveness of $\text{C}^2$RM on improving the energy efficiency of a FEEL system.
\end{abstract}
\section{Introduction}
Realizing the vision of exploiting the enormous data distributed at edge devices (e.g., smartphones and sensors) to train \emph{artificial intelligence} (AI) models has been pushing machine learning from the cloud to the network edge, called \emph{edge learning} \cite{gxzhu_2018_edge_learning}. Currently, arguably the most popular edge learning framework is federated learning that preserves users’ privacy by distributed learning over devices \cite{gxzhu2018FEEL, deniz2019federated_edge_learning, sqwang_2018_edge_learning}. Each \emph{round} of the iterative learning process involves the broadcasting of a global model to devices, their uploading of local model updates computed using local datasets, and a server’s aggregation of the received local updates for updating the global model. Executing complex learning tasks on energy and resource constrained devices is a main challenge faced in the implementation of edge learning. To tackle the challenge, the next-generation mobile \emph{systems-on-chip} (SoC) will feature a heterogeneous architecture comprising a \emph{central processing unit} (CPU) and a \emph{graphics processing unit} (GPU) [and even a \emph{digital processing unit} (DSP) in some designs] \cite{R9}. Experiments have demonstrated its advantages in terms of performance and energy efficiency. In this work, we address the issue of energy-efficient implementation of \emph{federated edge learning} (FEEL) in a wireless system comprising next-generation devices capable of heterogeneous computing. To this end, a novel framework is proposed for energy-efficient joint management of \emph{computation-and-communication} ($\text{C}^2$) resources at devices. 

\subsection{Edge Implementation of Federated Learning}
The implementation of FEEL in wireless networks faces two key challenges among others. As mentioned, the distributed learning process requires potentially many edge devices to periodically upload high-dimensional local model updates to an edge server. This includes the first challenge that the excessive communication overhead generated in the learning process can overwhelm the air interface that has finite radio resources but needs to support FEEL as well as other services. Designing techniques for suppressing the overhead forms a vein of active research, called \emph{communication-efficient FEEL}. Diversified approaches have been proposed such as device scheduling \cite{nishio2018background,yang2018background}, designing customized multi-access technologies \cite{gxzhu2018FEEL,deniz2019federated_edge_learning}, and optimizing uploading frequencies \cite{sqwang_2018_edge_learning}.

The second challenge faced by FEEL is how to execute power hungry learning tasks (e.g., training AI models each typically comprising millions of parameters) on energy constrained devices. Tackling the challenge by designing energy-efficient techniques leads to the emergence of an active research theme, called \emph{energy-efficient FEEL}, which is the topic of current investigation. While there exists a rich literature of energy-efficient techniques for \emph{resource management} (RM) in radio access networks (see e.g., \cite{liye}), the designs of their counterparts for FEEL, which are the main research focus in energy-efficient FEEL, are different due to the changes on the system objective and operations. In particular, FEEL systems aim at improving the learning performance instead of providing a radio access service and performing update aggregation instead of decoupling multiuser data. In early works \cite{QS,chen2019joint,sun2019energyaware}, researchers proposed radio-RM techniques (e.g., bandwidth allocation and device scheduling) to improve the tradeoff between devices’ transmission energy consumption and learning performance. More recent research accounts for computation energy consumption which usually constitutes a substantial part of a device’s total consumption in the learning process, motivating the design of $\text{C}^2$RM techniques in \cite{tran2019energy,yang2019energy,mo2020energyefficient}. In \cite{tran2019energy}, the tradeoffs between learning performance (in terms of accuracy and latency) and devices' energy consumption are optimized by balancing the communication and computation latencies under a total latency constraint, referred hereafter as \emph{$\text{C}^2$ time division}. In a standard FEEL design such as that in \cite{tran2019energy}, the updates by devices are usually assumed synchronized. It arises from the operation of gradient aggregation and refer to the requirement that all local updates need to be received by the server before the global model can be updated. Consequently, all devices are allowed the same duration (per-round latency) for uploading their local gradients and thus synchronized in update transmission. The assumption on synchronized updates is relaxed in \cite{yang2019energy} where the learning latency is measured by a weighted sum of individual devices’ heterogeneous latencies and then the learning-energy tradeoffs similar to those in \cite{tran2019energy} are optimized. On the other hand, the idea of clock frequency control, or called \emph{dynamic voltage and frequency scaling} (DVFS) \cite{R15}, was explored in \cite{mo2020energyefficient} for energy-efficient FEEL based on the assumption that each device has a CPU featuring DVFS. Then the $\text{C}^2$RM was optimized where the communication RM is based on either \emph{time-division multiple access} (TDMA) or \emph{non-orthogonal multiple access} (NOMA) and the computation RM is based on dynamic frequency scaling. While prior work assumes single-processor devices, CPU-GPU heterogeneous computing discussed in the sequel is a new paradigm of mobile computing supporting applications with intensive data crunching. The area of energy-efficient FEEL based on heterogeneous computing is uncharted and explored in this work. 
\subsection{CPU-GPU Heterogeneous Computing}
The mentioned next-generation heterogeneous SoC are capable of supporting diversified workloads such as communication, signal processing, inference, and learning, which arise from a wide range of new mobile applications. Examples of such chips include Snapdragon by Qualcomm, R-series by AMD, and Kirin 970 by Huawei. Via the cooperation of the CPU and GPU on the same chip for executing a single task, \emph{heterogeneous computing} on such SoC fully utilizes the computation resources of both processors to optimize the computation performance and energy efficiency \cite{R16}. In particular, the new paradigm has been demonstrated by experiments to accelerate the running of deep neural networks on smartphones \cite{R8, R18}. Existing research on heterogeneous computing focuses on several main design issues including workload partitioning \cite{R9, R11, R14, R15}, dynamic frequency scaling \cite{R5, R20}, and memory access scheduling \cite{R10}. The first two issues are addressed in this work. Specifically, workload partitioning refers to dividing and allocating workload of a task over the integrated CPU and GPU according to the task requirements and the processors' states (e.g., temperatures) and characteristics (e.g., speeds and energy efficiencies) so as to optimize the overall performance and efficiency \cite{R9, R11, R14, R15}. On the other hand, frequency scaling (or DVFS) previously considered for a single CPU \cite{tran2019energy,mo2020energyefficient} becomes the more sophisticated management in heterogeneous computing due to the joint CPU-GPU control \cite{R20}. While prior work on heterogeneous computing focuses on a single device, we study the joint control of workload partitioning and DVFS in the context of a large system comprising multiple devices and furthermore explore their integration with radio-RM. 

\subsection{Contributions}
The objective of this work is not to contribute any new learning technique but to focus on $\text{C}^2$RM to facilitate the implementation of a standard FEEL technique in a wireless network. To this end, we consider a FEEL system consisting of one edge server and multiple edge devices. Its main difference from those in existing work (see e.g., \cite{yang2019energy}) is that each device is equipped with a CPU-GPU platform enabling heterogeneous computing. The design objective is to minimize the sum energy consumption at devices under a guarantee on learning performance (latency and accuracy) by jointly controlling $\text{C}^2$RM in the following four dimensions: 
\begin{itemize}
    \item[1) ] \emph{Bandwidth allocation}: Allocating bandwidths to devices for transmission of local model updates under a constraint on the total uplink bandwidth; 
    \item[2) ] \emph{$\text{C}^2$ time division}: Dividing the allowed per-round latency for the computation and communication of each device;
    \item[3) ] \emph{CPU-GPU workload partitioning}: Partitioning and allocating the computation workload for each device to its CPU and GPU; 
    \item[4) ] \emph{CPU-GPU frequency scaling}: Controlling the CPU-GPU frequencies/speeds at each device. 
\end{itemize}
While the RM control in the first two dimensions are considered in prior work as discussed, the other two are unique for heterogeneous computing. Their joint control over multiple devices is a challenging and open problem. To the best of the authors’ knowledge, this work represents the first attempt on studying heterogeneous computing in the context of energy-efficient FEEL. The main contributions are described as follows. 
\begin{itemize}
    \item \textbf{Equilibrium based $\text{C}^2$RM framework:} It is mathematically proved that the control policies are optimal if and only if they achieve the following equilibriums:
    \begin{itemize}
        \item[1) ] The CPU-GPU pair of every device has \emph{equal energy-workload rates}, defined as the increase in energy consumption per additional workload for each processing unit; 
        \item[2) ] A similar equilibrium also exists with respect to the processing units' computation speeds as their optimal values are proved to be proportional to the corresponding workloads.
        \item[3) ] All devices have \emph{equal energy-time rates}, defined as the energy rates with respect to the communication/computation latency; 
        \item[4) ] All devices have \emph{equal energy-bandwidth rates}, defined as the energy rates with respect to the allocated bandwidth.
    \end{itemize}
The equilibriums are applied to obtaining the optimal polices for 1) computation RM, 2) communication RM, and 3) joint $\text{C}^2$RM. They are either derived in closed-form or computed using low-complexity algorithms. In particular, for low-complexity joint $\text{C}^2$RM, the problem is decomposed into one master problem for achieving the equilibrium in energy-time rates, and two sub-problems for separately achieving the two equilibriums in energy-workload rates and energy-bandwidth rates. The resultant algorithm is shown to have much lower complexity than a standard solution method such as \emph{block coordinate descent} (BCD).
\item \textbf{$\text{C}^2$ aware scheduling:} Building on the equilibrium framework, an energy-efficient scheduling scheme is designed to select a fixed number of devices for participating in FEEL. The novelty lies in the design of \emph{$\text{C}^2$ aware scheduling metric} that balances both the channel state and computation capacity of each device. Specifically, given the objective of sum energy minimization, the metric is designed to be the energy consumption of a device given the derived optimal $\text{C}^2$ time division and equal bandwidth allocation. In addition, we analyze the effect of the number of scheduled devices on model convergence. 
\item \textbf{Greedy spectrum sharing:} For the preceding designs, we assume fixed bandwidth allocation in each round of the iterative learning process. Consequently, the heterogeneous computation latencies of devices generates “spectrum holes” (unused spectrum-time blocks) that can be scavenged for further improving the energy efficiency. To this end, the scheme of greedy spectrum sharing is designed featuring a novel metric, called \emph{energy-bandwidth acceleration rate} and defined as the derivative of the energy-bandwidth rate, for selecting an available device to transmit using the spectrum hole. A larger value of the metric, implies steeper energy reduction for a device when it is allocated additional bandwidth.
\end{itemize}

\begin{figure}[t]
    \centering
    \includegraphics[width=0.95\textwidth]{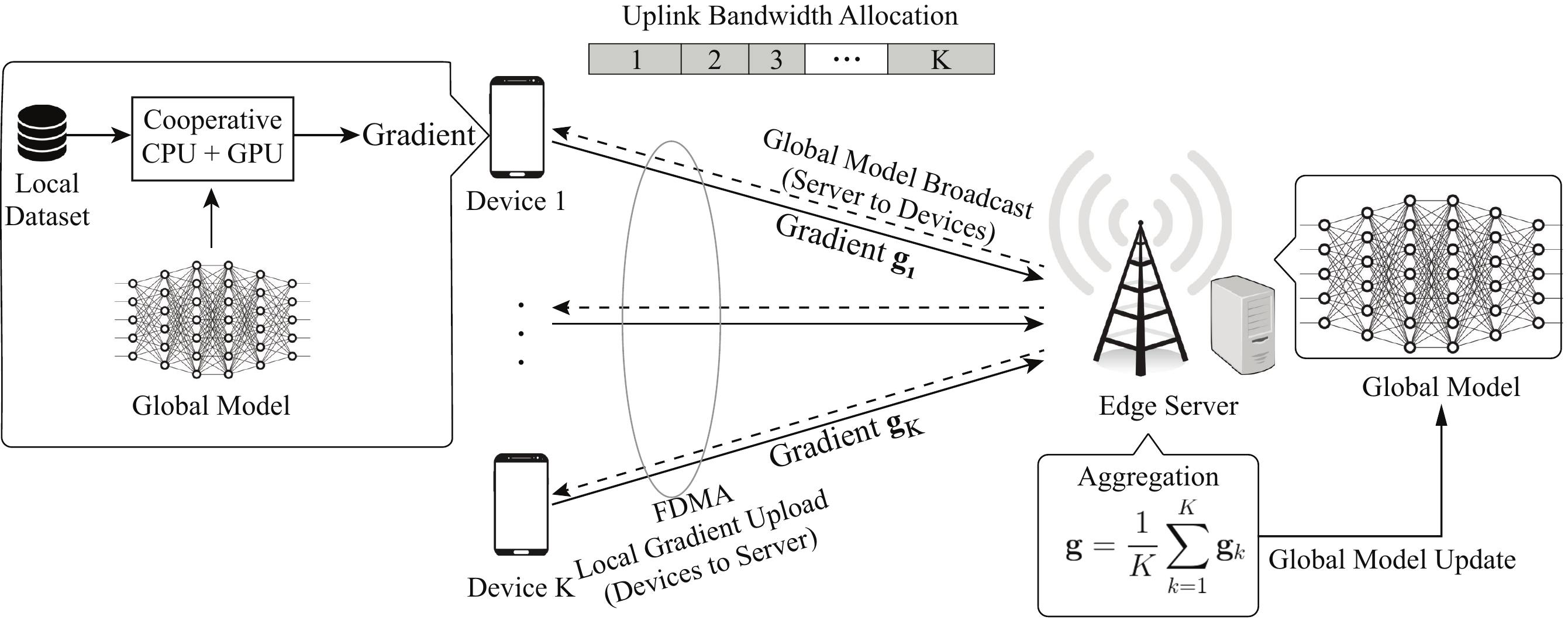}
    \caption{A wireless system supporting FEEL over devices capable of CPU-GPU heterogeneous computing.}
      \label{Fig:system_model}
\end{figure}

\section{Models and Metrics}\label{system model}
Consider the FEEL system in Fig.~\ref{Fig:system_model} comprising a single edge server and $K$ edge devices, denoted by an index set $\mathcal{K}=\{1,\cdots,K\}$. Each iteration between local gradient uploading and global model updating is called a \emph{communication round}, or \emph{round} for short. It is assumed that the server has perfect knowledge of the  multiuser channel gains and local computation characteristics, which can be obtained by feedback. Using this information, the server determines the energy-efficient strategies for device scheduling and $\text{C}^2$RM. 

\subsection{Federated Learning Model}
A standard federated learning technique (see e.g.,~\cite{federatedlearning}) is considered as follows. A global model, represented by the parameter set $\mathbf{w}$, is trained collaboratively across the edge devices by leveraging local labelled datasets. For device $k$, let $\mathcal{D}_k$ denote the local dataset and define the local loss function as
   $ F_k(\mathbf{w})=\frac{1}{|\mathcal{D}_k|}\sum_{(\mathbf{x}_j,y_j)\in\mathcal{D}_k}\ell(\mathbf{w};\mathbf{x}_j,y_j),$
where $\ell(\mathbf{w};\mathbf{x}_j,y_j)$ is the sample-wise loss function quantifying the prediction error of the model $\mathbf{w}$ on the training sample $\mathbf{x}_j$ with regard to its label $y_j$. For convenience, we assume a uniform size for local datasets: $|\mathcal{D}_k|\equiv D$, for all $k$. Then the global loss function on all the distributed datasets can be written as
\begin{equation}
    F(\mathbf{w})=\frac{\sum_{(\mathbf{x}_j,y_j)\in\cup_k\mathcal{D}_k}\ell(\mathbf{w};\mathbf{x}_j,y_j)}{|\cup_k\mathcal{D}_k|}=\frac{1}{K}\sum_{k=1}^KF_k(\mathbf{w}).
\end{equation}
The learning process is to minimize $F(\mathbf{w})$, that is, $\mathbf{w}^*=\arg\min F(\mathbf{w})$. For ease of exposition, we focus on the gradient-averaging implementation while the current designs also apply to the model-averaging implementation~\cite{federatedlearning}. In each round, say the $i$-th round, the server broadcasts the current model ${\mathbf{w}}^{(i)}$ as well as selection indicators $\{\rho_k\}$ to all edge devices, where the indicator $\rho_k=1$ if device $k$ is scheduled, or $0$ otherwise. Suppose $M$ devices are scheduled for participation in each round, denoted by an index set $\mathcal{M}_i$ for the $i$-th round. Based on the received model ${\mathbf{w}}^{(i)}$, each scheduled device calculates the gradient $\nabla F_k({\mathbf{w}}^{(i)})$ using its local dataset. Upon completion, the local gradients are transmitted to the server for aggregation:
\begin{equation}\label{gradient-aggregation}
    {\mathbf{g}}^{(i)}=\frac{1}{M}\sum_{k\in\mathcal{M}_i}\nabla F_k({\mathbf{w}}^{(i)})
\end{equation}
Synchronized updates by devices are assumed. The global model is then updated by \emph{stochastic gradient descent} (SGD) as 
    ${\mathbf{w}}^{(i+1)}={\mathbf{w}}^{(i)}-\eta {\mathbf{g}}^{(i)}$,
where $\eta$ is the learning rate. The process iterates until the model converges. 

\subsection{Model of Heterogeneous Computing}

\subsubsection{Workload model} Adopting a standard model (see e.g.,~\cite{zhang2018flops}), the total workload $W$ for a computation task is given as
$W=N_{\sf FLOP}\times D$,
where $D$ is the local dataset size and $N_{\sf FLOP}$ denotes the number of \emph{floating point operations} (FLOPs) needed for processing each sample. Furthermore, we define $f_k^{(c)}$ and ${f'}_k^{(c)}$ (in cycle/s) as the clock frequency of the CPU and GPU at device $k$, respectively. It follows that the computing speeds of CPU and GPU can be defined as $f_k = f_k^{(c)}\times n_k$ and $f'_k = {f'}_k^{(c)}\times n_k'$ with $n_k$ and $n_k'$ denoting the number of CPU and GPU FLOPs per cycle, respectively.

\subsubsection{Workload partitioning model}
For workload partitioning in local gradient computation, we consider input-sample partitioning, also known as \emph{data parallelism}~\cite{partition,parallel}. As a result, each processor (CPU or GPU) processes a fraction of input samples. The partitioned workloads for the CPU and GPU at device $k$ are denoted as $W_k$ and ${W'}_k$ respectively, with
\begin{equation}\label{InEq: workload_constrain}
    \text{(Workload constraint)}\quad W_k+{W}_k'=W,\quad \forall k\in\mathcal K.
\end{equation}
Given the partitioned workloads and CPU/GPU frequencies, the local computation time for the total workload $W$, denoted as ${t'}_k$, is given as
\begin{equation}\label{InEq: local computation time}
    \text{(Local computation time)}\quad {t}_k'=\max\left\{\frac{W_k}{f_k},\frac{{W}_k'}{{f}_k'}\right\},\quad \forall k\in\mathcal K.
\end{equation}


\subsubsection{DVFS model}
For a CMOS circuit, the power consumption of a processor can be modeled as a function of clock frequency: $P=\Psi \left[f^{(c)}\right]^3$ with the coefficient $\Psi$ [in $\text{Watt}/ (\text{cycle/s})^{3}$] depending on the chip architecture and $f^{(c)}$ being the clock frequency~\cite{liu2012dvfs}. Using this model, the power consumption of CPU and GPU at device $k$ can be written as 
\begin{equation}\label{cpu-gpu-power}
    P_k^{\text{CPU}}=\Psi_k^{\text{CPU}}\l(f^{(c)}_k\r)^3 = C_kf_k^3 \qquad\text{and}\qquad P_k^{\text{GPU}}=\Psi_k^{\text{GPU}}\l({f'}_k^{(c)}\r)^3 = G_k{f'}_k^3,
\end{equation}
where $C_k=\Psi_k^{\text{CPU}}/n_k^3$ and $G_k=\Psi_k^{\text{GPU}}/{n'}_k^3$. Using these functions, frequency scaling refers to controlling the power of CPU and GPU by adjusting their speeds ${f}_k$ and ${f}_k'$, respectively.

\begin{remark}
\emph{(CPU/GPU Computation Efficiency). The coefficient $C_k$ (or $G_k$) in \eqref{cpu-gpu-power} characterizes the \emph{computation efficiency} of a CPU (or GPU), defined as the rate of power growth in response to the increase of the cubed computing speeds. In practical data processing based on heterogeneous computing, the GPU tends to play a main role (with $G_k<C_k$) while the CPU contributes supplementary computing resources.
}
\end{remark}

\subsubsection{Energy consumption model}
Given the time durations $W_k/f_k$ and $W_k'/f_k'$ for CPU and GPU to complete their tasks with the workloads $W_k$ and $W'_k$, the resultant energy consumption at device $k$ can be written as 
\begin{equation}
    E_k^{\text{CPU}}=C_kW_kf_k^2\qquad\text{and}\qquad E_k^{\text{GPU}}=G_kW_k'{f'}_k^2.
\end{equation}
Then the total computation energy consumption of device $k$ per round is 
\begin{equation}\label{Eq: E_cmp}
    E_k^{\text{cmp}}=C_kW_kf_k^2+G_k{W}_k'{f'}_k^2,\quad \forall k\in\mathcal{K}.
\end{equation}

\subsection{Multiple-Access Model}
Consider the \emph{frequency-division multiple access} (FDMA) for gradient uploading with the total bandwidth $B$. Let $B_k$ denote the allocated bandwidth for device $k$ in an arbitrary round, which is fixed throughout the round. This assumption is relaxed in Section~\ref{Extension}. Then we have the following constraint:
\begin{equation}\label{Eq: bandwidth_constraint}
 \text{(Bandwidth constraint)}\quad \sum_{k=1}^KB_k=B.
\end{equation}
Channels are assumed to be frequency non-selective. The edge server usually recruits idling devices as workers that either static or at most moving at pedestrian speeds \cite{federatedlearning}. For this reason, we adopt the model of slow block fading. To be specific, the channel gain of device $k$, denoted as $h_k$, is assumed to remain unchanged within one round but varies \emph{independently and identically} (i.i.d.) over rounds. Given synchronous updates, a time constraint is set for each round:
\begin{equation}\label{InEq: time_constraint}
    \text{(Latency constraint)}\quad {t}_k'+t_k\leq T,\quad \forall k\in\mathcal K,
\end{equation}
where ${t}_k'$ and $t_k$ denote the time for local training/computation and gradient uploading of device $k$, respectively; $T$ is the maximum total time for one round. A prerequisite for a scheduled device is that it can meet the above time constraint. Let $P_k^{\text{TX}}$ denote the transmission power of device $k$. The achievable rate, denoted by $r_k$, can be written as
\begin{align}\label{Eq:shannon}
    r_k=B_k\log\left(1+\frac{P_k^{\text{TX}}h_k^2}{N_0}\right),\quad \forall k\in\mathcal{K},
\end{align}
where $N_0$ is the spectrum density of the complex white Gaussian channel noise. Let $L=|\mathbf{g}|\alpha$ denote the gradient size (in bit) with $\alpha$ denoting a sufficient large number of bits for quantizing each parameter with negligible distortion. Then the data rate is $r_k = L/t_k,\forall k\in\mathcal{K}$.
By combining the result and \eqref{Eq:shannon}, the communication energy consumption of device $k$ per round is
\begin{equation}\label{Eq:E_upload}
E_k^{\text{cmm}}=B_kP_k^{\text{TX}}t_k = \frac{B_kt_kN_0}{h_k^2}\left(2^{\frac{L}{B_kt_k}}-1\right),\quad \forall k\in\mathcal{K}. 
\end{equation}

\subsection{Performance Metric}
It is assumed that the data distributed at devices satisfy the condition for model convergence within a finite number of rounds, denoted as $N$ (see e.g., \cite{signSGD}). As implied by (\ref{Eq:shannon}), capacity achieving codes are deployed to ensure reliable transmissions. Consequently, wireless channels have no effects on $N$ though they affect per-round latencies and energy consumption. The objective of energy-efficient RM is to minimize the total energy consumption of all active devices in the $N$-round learning process, namely $\sum_{i=1}^N\sum_{k=1}^K\left(E_{k}^{\text{cmp}}[i]+E_{k}^{\text{cmm}}[i]\right)$ with $E_k^{\text{cmp}}$ and $E_{k}^{\text{cmm}}$ given in (\ref{Eq: E_cmp}) and (\ref{Eq:E_upload}), respectively. Given block fading and the per-round latency constraint in (\ref{InEq: time_constraint}), the objective can be straightforwardly proved to be equivalent to minimize the total energy consumption for each round. This is aligned with a standard approach in adaptive transmission and its proof is similar to those in the literature (see Lemma 1 in \cite{wen2020joint}) and thus omitted for brevity. It follows that for subsequent designs, it is sufficient to consider an arbitrary round and the corresponding total energy consumption, i.e.,  $\sum_{k=1}^K\left(E_k^{\text{cmp}}+E_k^{\text{cmm}}\right)$, as the performance metric, is called \emph{sum energy} hereafter.
\section{Energy-Efficient $\text{C}^2$ Resource Management}
In this section, considering the case that all devices are scheduled for uploading, we study: 1) computation RM, 2) communication RM, and 3) joint $\text{C}^2$RM by analyzing and computing the optimal policies. Then, the results are applied to establishing an energy-learning tradeoff. 
\subsection{Computation Resource Management}
\subsubsection{Problem Formulation}
The computation resources distributed at devices can be managed by controlling their workload partitioning and speed scaling to minimize the sum computation energy. The design is formulated as the following optimization problem.
\begin{equation*}\label{P:computation}{\bf (P1)}\quad
\begin{aligned}
    \min_{\{W_k,f_k,f_k'\}}\quad&\sum_{k=1}^K\left(C_kW_kf_k^2+G_kW_k'f_k^{'2}\right)\\
    \text{s.t.}\quad\quad&W_k+{W'}_k=W,\quad W_k\geq 0,\quad {W'}_k\geq 0,\quad \forall k\in\mathcal{K},\\
    &{t}_k'=\max\left\{\frac{W_k}{f_k},\frac{W'_k}{f_k'}\right\}, \quad \forall k\in\mathcal{K},\\
    &0\leq{t}_k'\leq {T'}_k, \quad \forall k\in\mathcal{K}.
\end{aligned}
\end{equation*}
where the three sets of constraints follow from \eqref{InEq: workload_constrain}, \eqref{InEq: local computation time}, and the time constraint with $T_k'$ being the allowed maximum computation time for device $k$.

\subsubsection{Optimal Policy}
The optimal policy for computation RM is derived in closed-form below.
\begin{lemma}\label{lemma: computation condition}
(Speed Scaling Rule). \emph{To minimize the sum computation energy, the computation speeds of CPU and GPU, $f_k$ and $f'_k$, should be scaled by following }
\begin{align}
   \frac{W_k}{f_k}=\frac{W_k'}{f_k'},\quad \forall k\in\mathcal{K},
\end{align}
\emph{given an arbitrary workload pair $(W_k, W_k^{'})$ with $W = W_k+W_k'$.}
\end{lemma}
\begin{proof}
    See Appendix \ref{proof: computation condition}.
\end{proof}

{The optimal strategy of workload-proportional speed scaling in Lemma 1 is intuitive and results from equalizing CPU and GPU computation time to avoid the slower one becoming the bottleneck of local gradient computation.}
Then, with Lemma~\ref{lemma: computation condition}, the original Problem {\bf P1} can be rewritten as follows:
\begin{equation*}\label{P:computation}{\bf (P2)}\quad
\begin{aligned}
    \min_{\{W_k,t_k'\}}\quad&\sum_{k=1}^K \frac{C_kW_k^3+G_k{W}_k^{'3}}{{t}_k^{'2}}\\
    \text{s.t.}\quad\quad&W_k+{W}_k'=W,\quad W_k\geq 0,\quad {W}_k'\geq 0,\quad \forall k\in\mathcal{K},\\
    &0\leq{t}_k'\leq {T}_k', \quad \forall k\in\mathcal{K}.
\end{aligned}
\end{equation*}

Next, it is found that the optimal policy  achieves an equilibrium between the CPU and GPU on each device, in terms of \emph{energy-workload rate}\footnote{To be mathematical rigor, we note that the notion of rate can only be defined by the derivative function. However, it can be proved that all the partial derivative functions appeared in this manuscript are equivalent to the derivative ones. Thereby, for consistency of notation, we define the rate with respect to partial derivative hereafter.} .
\begin{lemma}\label{lemma: energy-load rate}
(Energy-Workload Rate Equilibrium). \emph{Consider the \emph{energy-workload rates} of CPU and GPU given as 
\begin{align}\label{Def: energy-load rate}
\frac{\partial E_k^{\text{CPU}}}{\partial W_k}=\frac{3C_kW_k^2}{{t}_k^{'2}} \quad \text{and}\quad \frac{\partial E_k^{\text{GPU}}}{\partial {W}_k'}=\frac{3G_k{W}_k^{'2}}{{t}_k^{'2}},\quad \forall k\in\mathcal{K},
\end{align}
where $E_k^{\text{CPU}} = \frac{C_kW_k^3}{{t}_k^{'2}}$ and $E_k^{\text{GPU}} = \frac{G_k{W}_k^{'3}}{{t}_k^{'2}}$ denote the computation energy of CPU and GPU, respectively. Then, the necessary and sufficient condition for optimal workload allocation is
\begin{align}\label{Equilibrium: energy-load rate}
    \frac{\partial E_k^{\text{CPU}}}{\partial W_k}=\frac{\partial E_k^{\text{GPU}}}{\partial W_k'},\quad \forall k\in\mathcal{K},
\end{align}
with $W_k+{W}_k'=W$.
}
\end{lemma}
\begin{proof}
    See Appendix \ref{proof: energy-load rate}.
\end{proof}
Furthermore, one can observe that the objective function in {\bf P2} is a non-increasing function in ${t'}_k,\forall k\in\mathcal{K}$. Therefore, the optimality requires to maximize the computation time of each device, resulting in ${t'}_k^*={T'}_k,\forall k\in\mathcal{K}$, which is independent of the workload partitioning. Using the result as well as Lemmas 1 and 2, the optimal computation RM policy is obtained as follows.
\begin{proposition}\label{proposition: computation load allocation}
(Optimal Computation RM). \emph{The optimal workload allocation is
\begin{align}\label{Eq: Optimal_Workload_Allocation}
\text{(Optimal Workload Allocation)}\quad  W_k^*=\frac{\sqrt{G_k}W}{\sqrt{C_k}+\sqrt{G_k}},\quad W_k'^*=\frac{\sqrt{C_k}W}{\sqrt{C_k}+\sqrt{G_k}}, \quad k\in\mathcal{K},
\end{align}
where $C_k$ and $G_k$ are computation coefficients for CPU and GPU, respectively.
Moreover, the optimal speed scaling for CPU and GPU is
\begin{align}\label{proposition: frequency scaling}
 \text{(Optimal Speed Scaling)}\quad    f_k^*= \frac{W_k^*}{{T'}_k} ,\quad f_k'^*=\frac{W_k'^*}{{T'}_k}, \quad k\in\mathcal{K},~\qquad\qquad\quad\quad
\end{align}
where $T_k'$ is the allowed maximum computation time for device $k$.
}
\end{proposition}


\begin{remark}
(Energy-Efficient CPU-GPU Heterogeneous Computing). \emph{According to Proposition~\ref{proposition: computation load allocation}, more workload tends to be allocated to the processor with a smaller computation coefficient indicating a higher computation efficiency (see Remark 1), thereby reducing the devices' energy consumption.}
\end{remark}
\subsection{Communication Resource Management}
\subsubsection{Problem Formulation}
The total radio resources are managed by controlling bandwidth allocation and transmission time to minimize the sum transmission energy. The corresponding optimization problem can be formulated as
\begin{equation}\label{opt_rrm}{\bf (P3)}\quad
\begin{aligned}
    \min_{\{B_k,t_k\}}\quad&\sum_{k=1}^K\frac{B_k t_kN_0}{h_k^2}\left(2^{\frac{L}{B_kt_k}}-1\right)\\
    \text{s.t.}\quad\quad&\sum_{k=1}^KB_k=B,\quad B_k\geq 0,\quad \forall k\in\mathcal{K},\\
    &0\leq t_k\leq T_k,\quad \forall k\in\mathcal{K}\nn.
\end{aligned}
\end{equation}

The problem has a standard structure in the literature of energy-efficient communication (see e.g., \cite{changsheng}). It is convex and can be solved using a standard algorithm such as BCD. In the sequel, we present an alternative solution method yielding an equilibrium property that is useful for low-complexity policy computation faster than conventional methods.



\subsubsection{Properties of Optimal Policies} As before, the objective of Problem {\bf P3} is observed to be a non-increasing function in $t_k$. Therefore, it is optimal to maximize the transmission time of each device, resulting in $t_k^*=T_k,\forall k\in\mathcal{K}$. Next, to derive the optimal bandwidth allocation strategy, a necessary and sufficient condition is given as follows.

\begin{lemma}\label{lemma: energy-bandwidth rate}
(Energy-Bandwidth Rate Equilibrium). \emph{The \emph{energy-bandwidth rate} as defined earlier is mathematically obtained as 
\begin{align}\label{eqn: define nu_k}
    \frac{\partial E_k^{\text{cmm}}}{\partial B_k}  = \frac{t_kN_0}{h_k^2}\left(2^{\frac{L}{B_kt_k}}-\frac{L\ln{2}}{B_kt_k}2^{\frac{L}{B_kt_k}}-1\right)<0 ,\quad  k\in\mathcal{K}.
\end{align}
The optimal bandwidth allocation equalizes the energy-bandwidth rates as
\begin{align}\label{Eq: energy-bandwidth rate}
     \frac{\partial E_1^{\text{cmm}}}{\partial B_1}=\frac{\partial E_2^{\text{cmm}}}{\partial B_2}=\cdots=\frac{\partial E_K^{\text{cmm}}}{\partial B_K}=-\nu^*,
\end{align}
where $\nu^*$ is a constant and $\sum_{k=1}^KB_k = B$.
}
\end{lemma}
\begin{proof}
    See Appendix \ref{proof: energy-bandwidth rate}.
\end{proof}

It can be observed from the above lemma that the increase of allocated bandwidth reduces the communication energy consumption of the device.
%
%
Then, the optimal communication RM policy directly follows from the energy-bandwidth rate equilibrium as stated below.
\begin{proposition}\label{lemma: communication bandwidth allocation}
(Optimal Bandwidth Allocation).
\emph{The optimal policy for bandwidth allocation is}
\begin{align}\label{Eq: Optimal Bandwidth Allocation}
    B_k^*=\frac{L\ln{2}}{T_k\left[1+\mathcal{W}_0\!\!\left(\frac{h_k^2\nu^*-T_kN_0}{T_kN_0e}\right)\right]},\quad  k\in\mathcal{K},
\end{align}
\emph{where $T_k$ is the allowed maximum transmission time for device $k$; $\mathcal{W}_0(\cdot)$ is the Lambert $W$ function (principal branch) and $e$ is the Euler's number.}
\end{proposition}



Proposition~\ref{lemma: communication bandwidth allocation} suggests that $B_k^*$ is a non-increasing function with respect to $h_k^2$. It means that more bandwidths should be allocated to devices with weaker channels for the benefit of sum communication energy reduction.


Obtaining the optimal bandwidths $\{B^*_k\}$ via~\eqref{Eq: Optimal Bandwidth Allocation} requires computing the optimal energy-bandwidth rate $\nu^*$ by solving the following equation:
\begin{equation}\label{eqn: sum B}
    \sum_{k=1}^K\frac{L\ln{2}}{T_k\left[1+\mathcal{W}_0\!\!\left(\frac{h_k^2\nu^*-T_kN_0}{T_kN_0e}\right)\right]}=B.
\end{equation}
Due to the intractability of the Lambert function $\mathcal{W}_0$ and the unknown range of $\nu^*$, the solution cannot be found efficiently using standard methods such as Newton-Raphson method and bi-section search. An alternative method, \emph{derivative-free optimization} (DFO), which requires no derivative, has too high complexity. To address this issue, a fast algorithm for policy computation is designed below.

\subsubsection{Optimal Policy Computation} Instead of solving $\nu^*$ and $\{B^*_k\}$ in a sequential order, the fast algorithm calculates the optimal values iteratively, comprising the following two phases.
\begin{itemize}
\item {\bf Phase I}: (Bandwidth Optimization). Given the energy-bandwidth rate $\nu$, compute the bandwidth allocation $\{B_k\}$ using~\eqref{Eq: Optimal Bandwidth Allocation}. As $\nu$ is not optimal, the computed $\{B_k\}$ may not satisfy the bandwidth constraint in~\eqref{Eq: bandwidth_constraint}.
Thus it is necessary to normalize each $B_k$ as $\widetilde{B}_k  = \frac{B}{\sum_{k=1}^KB_k}B_k$. 
\item {\bf Phase II}: (Energy-Bandwidth Rate Updating). Calculate the respective energy-bandwidth rates $ \nu_k= - \frac{\partial E_k^{\text{cmm}}}{\partial B_k}$, $\forall k\in\mathcal{K}$, by substituting $\{\widetilde{B}_k\}$ and $\{t_k = T_k\}$ into~\eqref{eqn: define nu_k}. Then, update the current energy-bandwidth rate $\nu$ using $\{\nu_k\}$ as elaborated in the sequel.
\end{itemize}

The two phases are iterated until convergence as indicated by the \emph{energy-bandwidth rate equilibrium} in~\eqref{Eq: energy-bandwidth rate}. 

A key step in the algorithm is the update rule for $\nu$ in Phase II. To derive the rule, a useful property is given as follows.
\begin{lemma}\label{Pre:updating_rule}
\emph{Given $ \nu_k(B_k)= - \frac{\partial E_k^{\text{cmm}}}{\partial B_k}$ in~\eqref{eqn: define nu_k} and $B_k(\nu)$ in~\eqref{Eq: Optimal Bandwidth Allocation}, $\forall k\in\mathcal{K}$, for the $i$-th iteration between Phase I and II , the following holds:
\begin{itemize}
\item If $\nu^{(i-1)}> \nu^*$, then 
	1) $\nu^* <\max\limits_{k}\{\nu_k^{(i)}\} < \nu^{(i-1)}$, and 
	2) $\text{sgn}\left(\sum_{k=1}^KB_k^{(i)}-B\right)= -1$;
\item If $\nu^{(i-1)}< \nu^*$, then 
	1) $\nu^{(i-1)}<\min\limits_{k}\{\nu_k^{(i)}\} < \nu^*$, and 
	2) $\text{sgn}\left(\sum_{k=1}^KB_k^{(i)}-B\right)= 1$;
\end{itemize}
where $\nu_k^{(i)}=\nu_k(\widetilde{B}_k^{(i)}) $ with $\widetilde{B}_k^{(i)}  = \frac{B}{\sum_{k=1}^KB_k^{(i)}}B_k^{(i)} $ given $B_k^{(i)}  = B_k(\nu^{(i-1)}),\forall k\in\mathcal{K}$.}
\end{lemma}
\begin{proof}
See Appendix \ref{proof: algorithm 1}.
\end{proof}

Essentially, by induction, the above Lemma~\ref{Pre:updating_rule} simply implies that if the initial point $\nu^{(0)}$ is greater than $\nu^{*}$, and the update rule adopted for $\nu$ is $\nu^{(i)} = \max\limits_{k}\{\nu_k^{(i)}\}$, the convergence of $\nu$ to the global optimal $\nu^*$ is guaranteed. This is also true for the case of $\nu^{(0)} < \nu^*$ by involving the update rule as $\nu^{(i)} = \min\limits_{k}\{\nu_k^{(i)}\}$. It is further noted that the condition of $\nu^{(0)} > \nu^*$ or $\nu^{(0)} < \nu^*$ can be determined by calculating $\text{sgn}\left(\sum_{k=1}^KB_k^{(1)}-B\right)$, where $B_k^{(1)}  = B_k(\nu^{(0)})$ follows from~\eqref{Eq: Optimal Bandwidth Allocation}. In summary, we have the following update rule for $\nu$ in Phase II:
\begin{align}
\text{(Update rule)} \quad \nu^{(i)} = \left\{ 
 \begin{array}{rcl}
 \max\limits_{k}\{\nu_k^{(i)}\}, &  & \text{sgn}\left(\sum_{k=1}^KB_k^{(1)}-B\right)= -1;\\
 \min\limits_{k}\{\nu_k^{(i)}\}, &  & \text{otherwise},
\end{array}
\right. 
\end{align}
where $\{\nu_k^{(i)}\}$ and $\{B_k^{(1)}\}$ are specified in Lemma~\ref{Pre:updating_rule}. Based on the above rule, the algorithm for optimal bandwidth allocation is summarized in Algorithm~\ref{Algorithm:solve_nu}.

\begin{algorithm} [h]
    \caption{Optimal Bandwidth Allocation}
    \label{Algorithm:solve_nu}
    \textbf{Input:} initial value of $\nu$.\\
    \textbf{Output:} optimal $\nu^*$ and $\{B^*_k\}$.\\
    \textbf{Calculate:} indicator $y=\text{sgn}\left(\sum_{k=1}^KB_k(\nu)-B\right)$.\\
    \textbf{Repeat:}
    \begin{itemize}
    \item Calculate $\{B_k\}$  by substituting $\nu$ into Proposition~\ref{lemma: communication bandwidth allocation};
    \item Normalize $\widetilde{B}_k=\frac{B}{\sum_{k=1}^KB_k}B_k$, $k\in\mathcal{K}$;
    \item Calculate $ \{\nu_k\}$ by substituting  $\{\widetilde{B}_k\}$ into~\eqref{eqn: define nu_k}, $k\in\mathcal{K}$;
    \item Update $\nu=\frac{1-y}{2}\max\limits_{k}\{\nu_k\}+\frac{1+y}{2}\min\limits_{k}\{\nu_k\}$;
    \end{itemize}
    \textbf{Until} $\text{var}\left[\{\nu_k\}\right]<\varepsilon$ (to equalize the energy-bandwidth rates).
\end{algorithm}

\begin{remark}
(Low-Complexity and Optimality). \emph{
The complexity of Algorithm~\ref{Algorithm:solve_nu} is $O(\log\frac{1}{\varepsilon})$ with $\varepsilon$ denoting the target accuracy. For comparison, the computation of the DFO method for solving $\nu^*$ and $\{B_k^*\}$ in a sequential order is $O(\frac{1}{\varepsilon^2})$, and that of the BCD algorithm for directly solving Problem {\bf P3} is $O(\frac{K}{\varepsilon}\log\frac{1}{\varepsilon})$, which are much higher than that of Algorithm \ref{Algorithm:solve_nu}. Furthermore, the optimality is guaranteed by Algorithm \ref{Algorithm:solve_nu} as indicated by Lemma~\ref{Pre:updating_rule}.}
\end{remark}
\subsection{Joint $\text{C}^2$ Resource Management}
\subsubsection{Problem Formulation}
Building on the preceding results, an energy-efficient $\text{C}^2$RM framework is designed by joint control of workload partitioning, $\text{C}^2$ time division, and bandwidth allocation to minimize sum energy. The optimization problem is formulated as
\begin{equation}\label{P:joint_optimization}{\bf (P4)}\quad
\begin{aligned}
    \min_{\{W_k,{t}_k',B_k,t_k\}}\quad&\sum_{k=1}^K\left[\frac{C_kW_k^3+G_k{W'}_k^3}{{t'}_k^2}+\frac{B_kt_kN_0}{h_k^2}\left(2^{\frac{L}{B_kt_k}}-1\right)\right]\\
    \text{s.t.}\quad\quad&W_k+{W}_k'=W,\quad W_k\geq 0,\quad {W}_k'\geq 0, \quad \forall k\in\mathcal{K},\\
    &\sum_{k=1}^KB_k=B,\quad B_k\geq 0, \quad \forall k\in\mathcal{K},\\
    &{t}_k'+t_k\leq T,\quad {t}_k'\geq 0,\quad t_k\geq 0, \quad \forall k\in\mathcal{K}\nn.
\end{aligned}
\end{equation}

\subsubsection{Properties of Optimal Policy} 
It is shown in the sequel that a set of equilibriums exist among devices in terms of energy rates with respect to different types of control variables when $\text{C}^2$RM is optimally energy-efficient. The insights facilitate designing low-complexity policy for solving Problem {\bf P4}. First, by the same argument as in Subsections A and B, the latency constraint in Problem {\bf P4} should be active for energy minimization: ${t'}_k+t_k=T,\forall k\in\mathcal{K}$. The optimal $\text{C}^2$ time division of $T$ has the following property.
\begin{lemma}\label{lemma: energy-time rate}
    (Energy-Time Rate Equilibrium). \emph{The \emph{energy-time rates} as defined earlier can be mathematically obtained as
    \begin{equation}\label{def: xi'}
        \xi_k':=\frac{\partial E_k^{\text{cmp}}}{\partial {t}_k'}=-\frac{2(C_kW_k^3+G_k{W'}_k^3)}{{t'}_k^3},\quad k\in\mathcal{K},
    \end{equation}
    and
    \begin{equation}\label{def: xi}
        \xi_k:=\frac{\partial E_k^{\text{cmm}}}{\partial t_k}=\frac{B_kN_0}{h_k^2}\left(2^{\frac{L}{B_kt_k}}-\frac{L\ln 2}{B_kt_k}2^{\frac{L}{B_kt_k}}-1\right),\quad k\in\mathcal{K}.
    \end{equation}
    The optimal $\text{C}^2$ time division for each device requires the energy-time rate equilibrium: 
    \begin{equation}
        \frac{\partial E_k^{\text{cmp}}}{\partial {t}_k'}=\frac{\partial E_k^{\text{cmm}}}{\partial t_k},\quad\forall k\in\mathcal{K},
    \end{equation}
    with ${t}_k'+t_k= T$.
}
\end{lemma}
The proof is similar to that for Lemma~\ref{lemma: energy-load rate} and thus omitted for brevity. 
Next, one can observe that Problem {\bf P4} integrates {\bf P2} and {\bf P3} by summing their objectives and combining their constraints. Therefore, {\bf P4} is also convex and furthermore the results in Lemmas \ref{lemma: energy-load rate} and \ref{lemma: energy-bandwidth rate} hold for the current case. Then, combining Lemmas \ref{lemma: energy-load rate}, \ref{lemma: energy-bandwidth rate} and \ref{lemma: energy-time rate} yields the following main result.

\begin{framed}
\vspace{-5mm}
\begin{theorem}\label{theorem: equilibrium condition}
(Equilibrium Based $\text{C}^2$RM). \emph{The optimal $\text{C}^2$RM policy achieves the following equilibriums:
\begin{itemize}
    \item[1)] For each device, the optimal $\text{C}^2$ time division equalizes the energy-computation-time and energy-communication-time rates:
    \begin{align}
        \frac{\partial E_k^{\text{cmp}}}{\partial {t}_k'}=\frac{\partial E_k^{\text{cmm}}}{\partial t_k},\quad \forall k\in\mathcal{K}.
    \end{align}
    \item[2)] The optimal bandwidth allocation equalizes the energy-bandwidth rates:
    \begin{align}
        \frac{\partial E_1^{\text{cmm}}}{\partial B_1}=\frac{\partial E_2^{\text{cmm}}}{\partial B_2}=\cdots=\frac{\partial E_K^{\text{cmm}}}{\partial B_K}.
    \end{align}
    \item[3)] The optimal workload allocation at each device equalizes the energy-workload rates:
    \begin{align}
        \frac{\partial E_k^{\text{CPU}}}{\partial W_k}=\frac{\partial E_k^{\text{GPU}}}{\partial W_k'}, \quad \forall k\in\mathcal{K}.
    \end{align}
    \item[4) ] The optimal speed scaling at each device equalizes the energy-speed rates:
    \begin{align}
        \frac{\partial E_k^{\text{CPU}}}{\partial f_k}=\frac{\partial E_k^{\text{GPU}}}{\partial f_k'}, \quad \forall k\in\mathcal{K}.
    \end{align}
\end{itemize}}
\end{theorem}
\vspace{-5mm}
\end{framed}

\begin{remark}\label{Remark: C2-Tradeoff}
(Equalizing $\text{C}^2$ Heterogeneity). \emph{In the process of synchronized updates, Theorem 1 suggests that energy-efficient $\text{C}^2$RM should equalize the heterogeneity in communication channels and computation efficiencies using the multi-dimensional control variables. First, the heterogeneity in multiuser channel states and CPU/GPU computation efficiencies at each device is equalized by bandwidth allocation and workload partitioning (with speed scaling) as reflected in the second and third (with fourth) equilibriums. Second, the heterogeneity in $\text{C}^2$ speeds is equalized by adjusting the $\text{C}^2$ time division according to the first equilibrium. It is worth mentioning that $\text{C}^2$ time division represents a \emph{$\text{C}^2$ tradeoff} that more communication resources can compensate for the lack of computation resources and vice versa.
}
\end{remark}

\subsubsection{Optimal Policy Computation} Though the optimal solution for Problem {\bf P4} can be obtained by finding the equilibriums in Theorem~\ref{theorem: equilibrium condition} numerically using the standard method of BCD, it has high complexity. To tackle this challenge, we design a more efficient algorithm as follows. First, Problem ${\bf P4}$ can be decomposed into one master problem and two sub-problems as follows:
\begin{framed}
\vspace{-1mm}
\begin{itemize}
    \item (Master Problem): The master problem is the optimization of $\text{C}^2$ time division. Denote ${E}^{\text{cmp}}(\{\widetilde T_k\})$ and ${E}^{\text{cmm}}(\{\widetilde T_k\})$ as the computation and communication energy, which are functions of the allowed maximum communication time $\{\widetilde T_k\}$ given optimized resource management strategies output by sub-problem 1 and sub-problem 2. Then, the master problem is cast as
    \begin{equation*}\label{master_problem}{\bf (MP)}\quad
    \begin{aligned}
        \min_{\{\widetilde T_k\}}\quad& {E}^{\text{cmp}}(\{\widetilde T_k\})+ {E}^{\text{cmm}}(\{\widetilde T_k\})\\
        \text{s.t.}\quad&0< \widetilde T_k < T,\quad \forall k\in\mathcal{K}.
    \end{aligned} 
    \end{equation*}
    \item (Two Sub-problems):
    \begin{itemize}
    \item ({Computation RM}): Problem ${\bf P1}$ with $\{{T}_k'=T-\widetilde T_k\}$;
    \item ({Communication RM}): Problem ${\bf P3}$ with $\{T_k = \widetilde T_k\}$;
    \end{itemize}
\end{itemize}
\vspace{-3mm}
\end{framed}

As Problem {\bf P4} is convex, the optimal solutions can be obtained via iterations between the master problem and two sub-problems. Each iteration comprises two steps: 1) given $\{\widetilde T_k\}$, compute the optimal RM strategies by solving two sub-problems; 2) given allocated bandwidths and partitioned workloads, calculate the sub-gradient of the master problem and apply it to updating $\{\widetilde T_k\}$ via the gradient descent method. The two steps are iterated until convergence. 


Note that the solution of one sub-problem, namely Problem {\bf P1}, can be calculated directly via~\eqref{Eq: Optimal_Workload_Allocation} and~\eqref{proposition: frequency scaling} while numerically calculation is required for solving the other sub-problem, namely Problem {\bf P3}. Given the efficient Algorithm~\ref{Algorithm:solve_nu} developed for computing $\{B^*_k\}$, in the sequel, we aim at further improving the efficiency by proposing a novel initialization method of $\nu$. To this end, another useful property is introduced in the following lemma.
\begin{lemma}
\emph{The optimal $\nu^*$ under the arbitrary given $\text{C}^2$ time division, namely $T'_k + T_k = T$ with $t'_k \leq T'_k$ and $t_k \leq T_k, \forall k \in \mathcal{K}$, can be calculated as 
\begin{equation}\label{nu}
    \nu^* =-\frac{1}{B}\sum_{k=1}^{K}T_k\frac{\partial E_k^{\text{cmm}}}{\partial t_k}\bigg|_{T_k} = -\frac{1}{B}\sum_{k=1}^{K}T_k\xi_k(T_k),
\end{equation}
where $\xi_k(T_k) = \frac{\partial E_k^{\text{cmm}}}{\partial t_k}\big|_{T_k}$ is the energy-communication time rate under current time division.}
\end{lemma}
The above result can be proved by noting $ T_k\frac{\partial E_k^{\text{cmm}}}{\partial t_k}\big|_{T_k}=B_k^*\frac{\partial E_k^{\text{cmm}}}{\partial B_k}\big|_{B_k^*} = -\nu^* B_k^*$ given the bandwidth constraint, namely $\sum_{k=1}^{K}B_k^* = B$. As the values of $\{\xi_k(T_k)\}$ are not attainable, $\nu^*$ can not be solved directly. Nevertheless, we can get a good initial point for $\nu$ by taking advantage of \eqref{nu}. To be specific, due to the fact that $\xi_k'(T'_k)$ and $\xi_k(T_k)$ converge with iterations and $\xi_k'(T'_k)$ can be calculated by solving {\bf P1} in closed-form, we propose to initialize $\nu$ as
\begin{equation}\label{eqn: initialize nu}
\nu_0:= -\frac{1}{B}\sum_{k=1}^K T_k\xi_k'(T'_k),
\end{equation}
Such initialization provides an increasingly better initial point that is closer to the optimal solution as the
outer iteration proceeds. The algorithm for computing the optimal $\text{C}^2$RM policy is summarized in Algorithm~\ref{Algorithm:joint_RM}. 

\begin{algorithm} [h]
    \caption{Optimal $\text{C}^2$RM}
    \label{Algorithm:joint_RM}
    \textbf{Input:} initial values of $\{\widetilde T_k\}$.\\
    \textbf{Output:} optimal $\{B_k^*\}$, $\{{t'}_k^* = T-\widetilde T^*_k\}$, $\{t_k^* = \widetilde T^*_k\}$.\\
    \textbf{Repeat:}
    \begin{itemize}
    \item Solve the two subproblems with $\{T_k = \widetilde T_k\}$ and $\{T_k' = T-\widetilde T_k\}$ as input:
    \begin{description}
        \item[$\triangleright$ Solve Problem ${\bf (P1)}$:] \hspace{3cm} - calculate $\{W_k\}$ and $\{W^{'}_k\}$ using~\eqref{Eq: Optimal_Workload_Allocation};
         \item[]  \hspace{3cm} - obtain $\{\xi_k'\}$ by substituting $\{t_k'=T_k'\}$ and $\{W_k, W^{'}_k\}$ into~\eqref{def: xi'};
        \item[$\triangleright$ Solve Problem ${\bf (P3)}$:] 
         \hspace{3cm} - calculate $\{B_k\}$  and solve $\nu$ using Algorithm~\ref{Algorithm:solve_nu} by involving
         \item[] \hspace{3.25cm}  the initialization method~\eqref{eqn: initialize nu};
        \item[] \hspace{3cm} - obtain $\{\xi_k\}$ by substituting $\{B_k\}$ and $\{t_k = T_k\}$ into~\eqref{def: xi};
    \end{description}
    \item Update the time division: $\widetilde T_k= \widetilde T_k-\eta(\xi_k-\xi_k'),\quad \forall k \in \mathcal{K}$;
    \end{itemize}
    \textbf{Until} $\|\boldsymbol{\xi}'-\boldsymbol{\xi}\|^2\leq\varepsilon$ (to equalize the energy-time rates), where $\boldsymbol{\xi}'\triangleq [\xi_1',\cdots,\xi_K']^{\top}$ and $\boldsymbol{\xi}\triangleq [\xi_1,\cdots,\xi_K]^{\top}$.
\end{algorithm}

\begin{remark}
(Low-Complexity Algorithm). \emph{The proposed Algorithm 2 is of low-complexity, which has complexity up to $O(K\log^2\frac{1}{\varepsilon})$ as compared to that of the conventional BCD method for solving ${\bf P4}$, which has complexity $O(\frac{K}{\varepsilon}\log\frac{1}{\varepsilon})$.}
\end{remark}

\subsection{Discussion on Energy-Learning Tradeoff}

In the literature, the model convergence of federated learning is characterized by either the averaged gradient norm (see e.g., \cite{signSGD}) or the loss function (see e.g., \cite{federatedlearning}), as a monotone decreasing function of the number of required rounds, $N$, and participating devices $K$. For example, it is reported in \cite{signSGD} that with a properly chosen learning rate (step size), the averaged gradient norm over rounds is shown to be proportional to $1/\sqrt{N}$ and $1/\sqrt{K}$. In existing work, the per-round latency, $T$, is assumed constant and thereby can be omitted in the convergence analysis. By considering $T$ or $N$ as a function of $\text{C}^2$ sum energy, denoted as $E_{\varSigma}$, we can discuss a learning-energy tradeoff as follows.

\begin{definition}
\emph{(\emph{Learning Latency}). The learning latency is defined as $T_{\sf total}=N \times T~(\text{in second})$.
}
\end{definition}


\subsubsection{Fixed number of participating devices, $K$} This implies a fixed distributed dataset and thus fixed $N$. On the other hand, it can be inferred from Lemma~\ref{lemma: energy-time rate} that $T$ is a monotone decreasing function of $E_{\varSigma}$. Thus, the learning latency can be written as $T_{\text{total}}=N\times T(E_{\varSigma})$, which establishes the tradeoff between $T_{\text{total}}$ and $E_{\varSigma}$. Specifically, reducing $E_{\varSigma}$ can result in increased learning latency, and vice versa. Based on Lemma \ref{lemma: energy-time rate}, one can see that the reduction in the learning latency leads to the increment in sum energy at a rate \emph{faster than linear}.
%
\subsubsection{Fixed per-round latency, $T$}  For this case, an energy-learning tradeoff also exists. To be specific, it can be easily shown that the number of devices that can satisfy the per-round latency constraint is a monotone increasing function of $E_{\varSigma}$. Thereby, reducing $E_{\varSigma}$ can result in a reduced number of devices participating in learning. This gives rise to a larger required number of rounds for convergence due to the shrinking of distributed dataset. Thus, the learning latency can be written as $T_{\text{total}}=N(E_{\varSigma})\times T$ with $N(E_{\varSigma})$ being a monotone decreasing function. Then the energy-learning tradeoff is one that shortening learning latency needs more energy and vice versa. Particularly, the intermediate parameter, namely the number of participating devices, controls the said tradeoff in this case. This further motivates us to explore $\text{C}^2$ aware device scheduling in the following section.

In general, the functions $T(E_{\varSigma})$ in the first case and $N(E_{\varSigma})$ in the second case have no closed form. The analysis of their properties (e.g., scaling laws) is an interesting topic but outside the scope of this work.
\section{$\text{C}^2$ Aware Scheduling}\label{section:device-scheduling}
When there are many devices providing more than sufficient data, it is desirable from the energy-efficiency perspective to select only a subset of devices for participating in model training. In this section, the scheduler design is presented and the effect of scheduling on model convergence is quantified.

\subsection{Scheduler Design}
The mathematical formulation of the design is similar to {\bf P4} by modifying the objective as $\sum\limits_{k=1}^K\rho_k\left[\frac{C_kW_k^3+G_k{W'}_k^3}{{t'}_k^2}+\frac{B_kt_kN_0}{h_k^2}\left(2^{\frac{L}{B_kt_k}}-1\right)\right]$ under the additional constraints, namely $\sum\limits_{k=1}^K\rho_k=M,~\rho_k\in\{0,1\}, ~\forall k\in\mathcal{K}$, where $\{\rho_k\}$ are the indicator variables for selecting devices.
This mixed integer non-linear problem is NP-hard and classical optimal algorithms (e.g., branch and bound) have too high complexity to be practical when $K$ is large. To tackle this challenge, we propose a novel $\text{C}^2$ aware scheduler design of low-complexity to minimize sum energy. 

We consider the problem of selecting $M$ out of $K$ devices for FEEL. The data distribution over devices is assumed i.i.d. as commonly make in the literature (see e.g., \cite{tran2019energy}). Then with $M$ fixed, device scheduling has no effect on the number of rounds $N$ required for learning (whose dependence on $M$ is studied in Section B). This allows scheduling decisions to be based on the $\text{C}^2$ states of devices, giving the name of $\text{C}^2$ aware scheduling. 



The key element of the scheduler is a scheduling metric designed as follows.
Based on the preceding analysis, it is desirable to select devices with good channels or/and good computation efficiencies for energy reduction. To identify such devices, we propose to first perform equal bandwidth allocation over all devices (i.e., $B_k=\bar{B}=B/K,~\forall k\in\mathcal{K}$) and then evaluate the resulting energy consumption of each device. Note that equal bandwidth allocation enables the high energy consumption of a device to indicate: 1) a poor channel, or 2) a low computation efficiency, or 3) both. 
Therefore, the devices' energy consumption with equal bandwidth allocation is a suitable scheduling metric given as:
\begin{equation}\label{Eq: Energy_Device_scheduling}
E_k =  \frac{\bar{B} t_k^*N_0}{h_k^2}\left(2^{\frac{L}{\bar{B}t_k^*}}-1\right) + \frac{a_kW^3}{(T-t_k^*)^2},\quad\forall k\in\mathcal{K},
\end{equation}
where $a_k\triangleq\frac{C_kG_k}{(\sqrt{C_k}+\sqrt{G_k})^2},\forall  k\in\mathcal{K}$. To compute its value,  the optimal transmission time $t^*_k$ is needed given $B_k=\bar{B}\triangleq B/K$. Based on the energy-time rate equilibrium developed in Lemma~\ref{lemma: energy-time rate}, $t_k^*$ solves the equation below:
\begin{equation}\label{eqn: one-shot}
   f(t_k^*) = \frac{\bar{B}N_0}{h_k^2}\left(2^{\frac{L}{\bar{B}t_k^*}}-\frac{L\ln{2}}{\bar{B}t_k^*}2^{\frac{L}{\bar{B}t_k^*}}-1\right) + \frac{2a_kW^3}{(T-t_k^*)^3} = 0.
\end{equation}
A closed-form expression for $t_k^*$ is hard to derive but can be numerically computed by a bi-section search since $f(t_k^*)$ is a monotonically increasing function. Given the values of $\{E_k\}$, the $\text{C}^2$ aware scheduler selects $M$ devices with the smallest values. The scheduling scheme is summarized in Algorithm \ref{algorithm: one-shot scheduling}.

\begin{algorithm} [h]
    \caption{$\text{C}^2$ Aware Scheduling}
    \label{algorithm: one-shot scheduling}
    \textbf{Initialization}: $\forall k\in\mathcal{K},~\rho_k=0,~\bar{B}=B/K$.\\
    \textbf{Output:} Subset $\mathcal{M}=\{k\in\mathcal{K}\mid \rho_k=1\}$.
    \begin{itemize}
        \item Solve for the optimal time divisions $\{t_k^*\}$ using \eqref{eqn: one-shot} and a bi-section search;
        \item Calculate the scheduling metrics $\{E_k\}$ by substituting $\{t_k^*\}$ into~\eqref{Eq: Energy_Device_scheduling};
        \item Select $M$ devices with the smallest $E_k$ and set their $\rho_k=1$.
    \end{itemize}
\end{algorithm}

%
\subsection{Effect of Scheduling on Convergence}
In this subsection, we quantify the effect of $\text{C}^2$ aware scheduling (without considering data importance) on the convergence rate of FEEL (in round) due to the reduced size of selected dataset for model updating. For tractable analysis, we follow the standard assumptions on the loss function as made in the literature (see e.g., \cite{sqwang_2018_edge_learning}).
\begin{assumption}\label{assumption: convexity and smoothness}
    \emph{It is assumed that the loss functions has the following two properties:}
    \begin{itemize}
        \item (Convexity and Smoothness). \emph{All functions $\{F_k\}$ are convex and $\beta$-smooth, that is, for all $\mathbf{u},\mathbf{v}$, we have
        \begin{align}
            \langle\nabla F_k(\mathbf{v}),\mathbf{u}-\mathbf{v}\rangle\leq F_k(\mathbf{u})-F_k(\mathbf{v})\leq \langle\nabla F_k(\mathbf{v}),\mathbf{u}-\mathbf{v}\rangle+\frac{\beta}{2}\|\mathbf{u}-\mathbf{v}\|^2
        \end{align}}
    \item (Variance Bound). \emph{Stochastic gradients $\{\nabla F_k(\mathbf{w})\}$ are unbiased and variance bounded by $\sigma^2$, that is, 
    \begin{align}
        \mathbb{E}[\nabla F_k(\mathbf{w})]=\nabla F(\mathbf{w})\qquad \text{and}\qquad \mathbb{E}[\|\nabla F_k(\mathbf{w})-\nabla F(\mathbf{w})\|^2]\leq \sigma^2,
    \end{align}
    where $\nabla F_k(\mathbf{w})$ and $\nabla F(\mathbf{w})$ denote the gradients of a local loss function and the global loss function, respectively, and the expectations are taken over all devices.}
    \end{itemize}
\end{assumption}
Under the above assumptions, the upper bound on the convergence rate of FEEL algorithm with device scheduling is given in the following theorem.
\begin{theorem}\label{theorem: convergence rate}
    (Convergence Rate with Scheduling). \emph{Given the learning rate as $\eta=\frac{1}{\sqrt{N}}$, the convergence rate of the FEEL algorithm with device-scheduling can be upper-bounded  as
    \begin{equation}\label{Learning Convergence Rate}
        \frac{1}{N}\sum_{i=0}^{N-1}\mathbb{E}_{\mathcal{M}_i}\left[F({\mathbf{w}}^{(i)})-F(\mathbf{w}^*)\right]\leq\frac{1}{\sqrt{N}}\left[\|\mathbf{w}^{(0)}-\mathbf{w}^*\|^2+\frac{\sigma^2(K-M)}{(K-1)M}\right].
    \end{equation}
    where $N$ denotes the number of communication rounds.}
\end{theorem}
\begin{proof}
    See Appendix \ref{proof: convergence rate}.
\end{proof}

Compared with the existing results without scheduling (see e.g., \cite{khaled2019analysis} and \cite{tran2019energy}), the last term in \eqref{Learning Convergence Rate} is new that characterizes the effect of scheduling on learning. One can observe from the term that increasing the number of scheduled devices $M$ leads to the vanishment of the biased term at a rate of $O(\frac{1}{M})$, giving rise to a faster convergence rate, however, at the cost of more sum energy consumption scales \emph{faster than linearly} with $M$.
\section{Extension: Greedy Spectrum Sharing}\label{Extension}
\begin{figure}[t!]
\centering
\includegraphics[width=11cm]{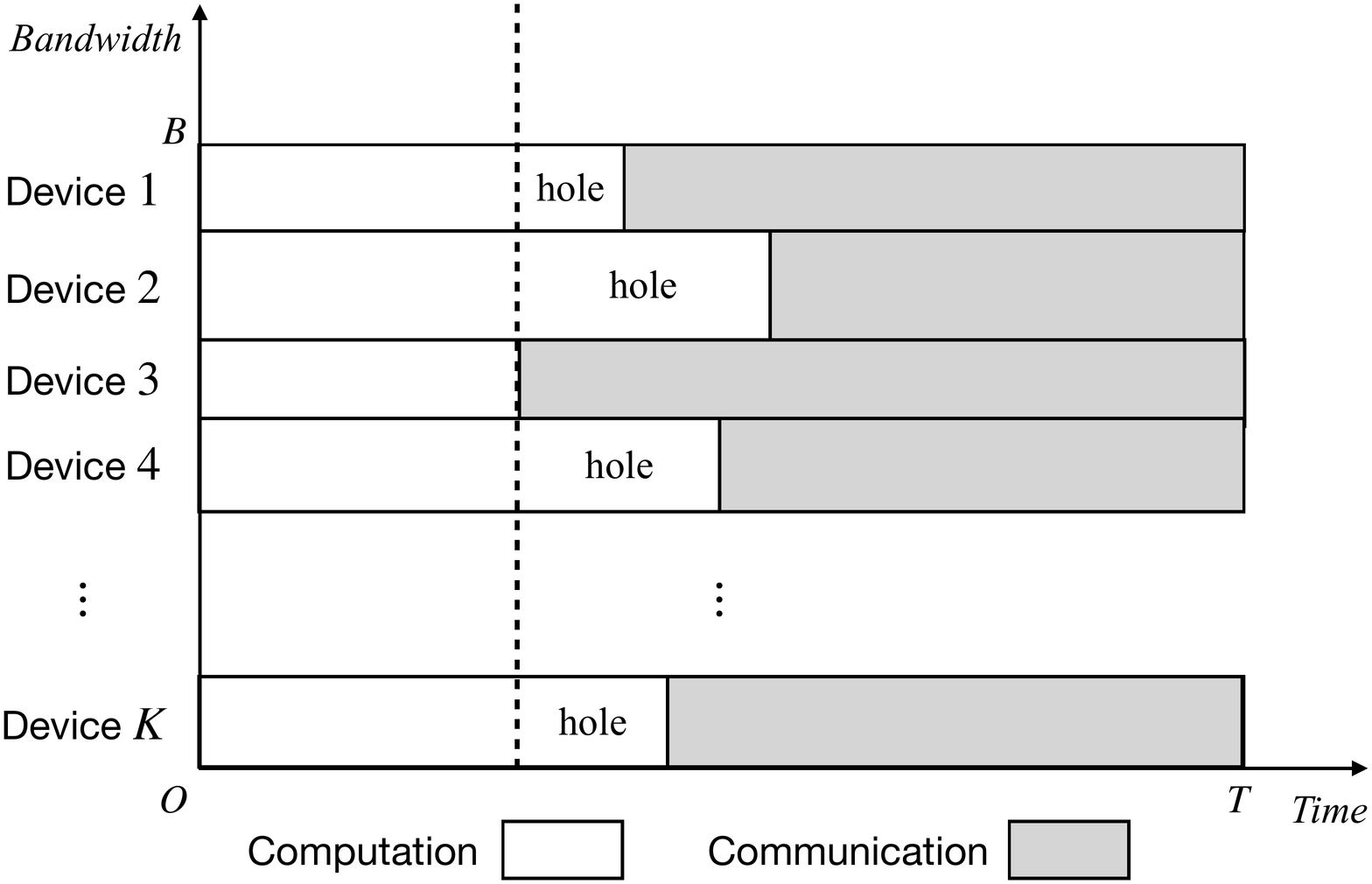}\vspace{-7mm}
\caption{Spectrum holes exist due to heterogeneous $\text{C}^2$ time division.}
\label{Fig: Spectrum Holes}
\end{figure}

In the preceding sections, the bandwidths are allocated at the beginning of each round and then fixed throughout the round. However, as mentioned, the heterogeneity in the computation durations of devices may result in spectrum holes as illustrated in Fig. \ref{Fig: Spectrum Holes}. Scavenging them by spectrum sharing among devices can reduce sum energy, which is the theme of this section. 
While optimizing the spectrum sharing is intractable, we design a practical scheme based on the greedy principle. The scheme allocates the spectrum hole upon its arrival to one of available devices for transmission (in addition to the pre-allocated bandwidth) so as to reduce its energy consumption, called \emph{greedy spectrum sharing}. To design the scheme, we divide the round into multiple time slots. Furthermore, let $\mathcal{S}$ and $\mathcal{S_{\ell}}$ denote the sets of scheduled devices and a subset of devices that have completed computation at the beginning of a particular time slot, respectively. Then, the resulting spectrum holes have the total bandwidth of $B-\sum_{k\in\mathcal{S}_{\ell}}B_k^*$ with $B_k^*$ denoting the pre-allocated bandwidth to device $k$, which is determined at the beginning of each round. Intuitively, all the unoccupied bandwidths should be allocated to the device $k\in \mathcal{S_{\ell}}$ that has the largest energy reduction with respect to bandwidth growth. Mathematically, this selection metric is to select the device, denoted as $k^*$, with the minimum energy-bandwidth rate at $B_k^*$ provided $ \frac{\partial E_k^{\text{cmm}}}{\partial B_k} < 0,\forall k\in\mathcal{K}$, that is 
\begin{align}
 k^* = \arg\min_{k\in \mathcal S_{\ell}}\frac{\partial E_k^{\text{cmm}}}{\partial B_k}\bigg|_{B_k^*}.
\end{align}
However, according to Lemma~\ref{lemma: energy-bandwidth rate}, the energy-bandwidth rates of all scheduled devices are equalized at equilibrium, making this criterion ineffective.  To address this issue, we design an alternative metric that can also effectively reduce sum energy. To this end, we introduce the following notion.
\begin{definition}\label{definition: acceleration}
(Acceleration Rate). \emph{The \emph{acceleration rate} is defined as the partial derivative of the energy-bandwidth rate:
     \begin{equation}\label{definition: phi}
     \phi_k(B_k):=\frac{\partial^2E_k^{\text{cmm}}}{\partial B_k^2}=\frac{2^{\frac{L}{B_kt_k}}L^2N_0(\ln 2)^2}{B_k^3t_kh_k^2} > 0,\quad \forall k\in\mathcal{K}.
    \end{equation}
}
\end{definition}

To be mathematical rigor, we note that the notion of acceleration rate can only be defined by the second derivative. However, it can be proved that the second-order partial derivative here is equivalent to the second derivative of energy w.r.t. the allocated bandwidth (see Appendix~\ref{proof: greedy spectrum sharing}).
A device with a small acceleration rate implies a lower energy-bandwidth rate upon being allocated with extra bandwidths, leading to more energy reduction. Therefore, we propose to adopt the criterion of minimum acceleration rate at $\{B_k = B_k^*\}$, namely $\{\phi_k(B_k^*)\}$, as follows
\begin{align}\label{Eq: device_slection}
\text{(Device Selection)}\quad k^*=\arg\min_{k\in\mathcal{S}_{\ell}}\phi_k(B_k^*).
\end{align}
Based on the above criterion, the scheme is summarized in Algorithm~\ref{algorithm: greedy spectrum sharing}.

\begin{algorithm} [h]
    \caption{Greedy Spectrum Sharing}
    \label{algorithm: greedy spectrum sharing}
    \textbf{Initialization}: Apply Algorithm~\ref{algorithm: one-shot scheduling} and~\ref{Algorithm:joint_RM} in sequential order to obtain $\{\rho_k^*,B_k^*\}$.\\
    Denote $\Delta t$ as the time slot duration and let $t_{\text{count}}=0$.\\
    For the subset of devices $\mathcal{S}=\{k\in\mathcal{K}\mid \rho_k=1\}$:\\
    \textbf{While} $t_{\text{count}}<T$:
    \begin{itemize}
        \item Denote $\mathcal{S}_{\ell}$ as the set of devices that have completed local computation at time $t_{\text{count}}$;
        \item For $k\notin \mathcal{S}_{\ell}$, no bandwidth will be occupied by them;
        \item For $k\in\mathcal{S}_{\ell}$, each device will first be allocated with bandwidth $B_k^*$;
        \begin{itemize}
            \item Calculate $\phi_k$ for each $k\in\mathcal{S}_{\ell}$ using~\eqref{definition: phi};
            \item Select the device via~\eqref{Eq: device_slection} and allocate all the accessable spectrum $B-\sum_{k\in\mathcal{S}_{\ell}}B_k^*$ to it;
        \end{itemize}
        \item $t_{\text{count}}=t_{\text{count}}+\Delta t$.
    \end{itemize}
\end{algorithm}
\section{Experimental results}\label{Simulation}
\subsection{Experiment Setup}
The simulation settings are as follows unless specified otherwise. In the FEEL system, there are $K=50$ devices and each of them is capable of CPU-GPU heterogeneous computing. The devices' CPU and GPU coefficients, $\{C_k\}$ and $\{G_k\}$, are uniformly selected from the the set $\{0.020,0.021,\cdots,0.040\}$ and $\{0.001,0.002,\cdots,0.010\}$, respectively.
Consider an FDMA system with the uplink bandwidth $B=5$ MHz. The channel gains $\{h_k\}$ are modeled as i.i.d. Rayleigh fading with average path loss set as $10^{-3}$. The noise variance is $N_0=10^{-9}$ W/Hz. The classification task aims at classifying handwritten digits from the well-known MNIST dataset. Each device is randomly assigned $20$ samples. The classifier model is implemented using a 6-layer \emph{convolutional neural network} (CNN) which consists of two $5\times 5$ convolution layers with ReLU activation, each followed by $2\times 2$ max pooling, a fully connected layer with 50 units and ReLU activation, and a final softmax output layer. The total number of parameters is $21,840$ and the computation workload is $W = 9.75$ MFLOPs. Furthermore, we suppose that each parameter of the training model gradient is quantized into 16 bits, and as a result, the transmission overhead is $L = 3.49\times 10^5$ bits.


\subsection{$\text{C}^2$ Resource Management}

\subsubsection{Energy-efficient RM} 
\begin{figure}[t!]
    \centering
    \subfigure[Uniform time division]{
    \label{Fig: uniform time division}
    \includegraphics[width=8.0cm]{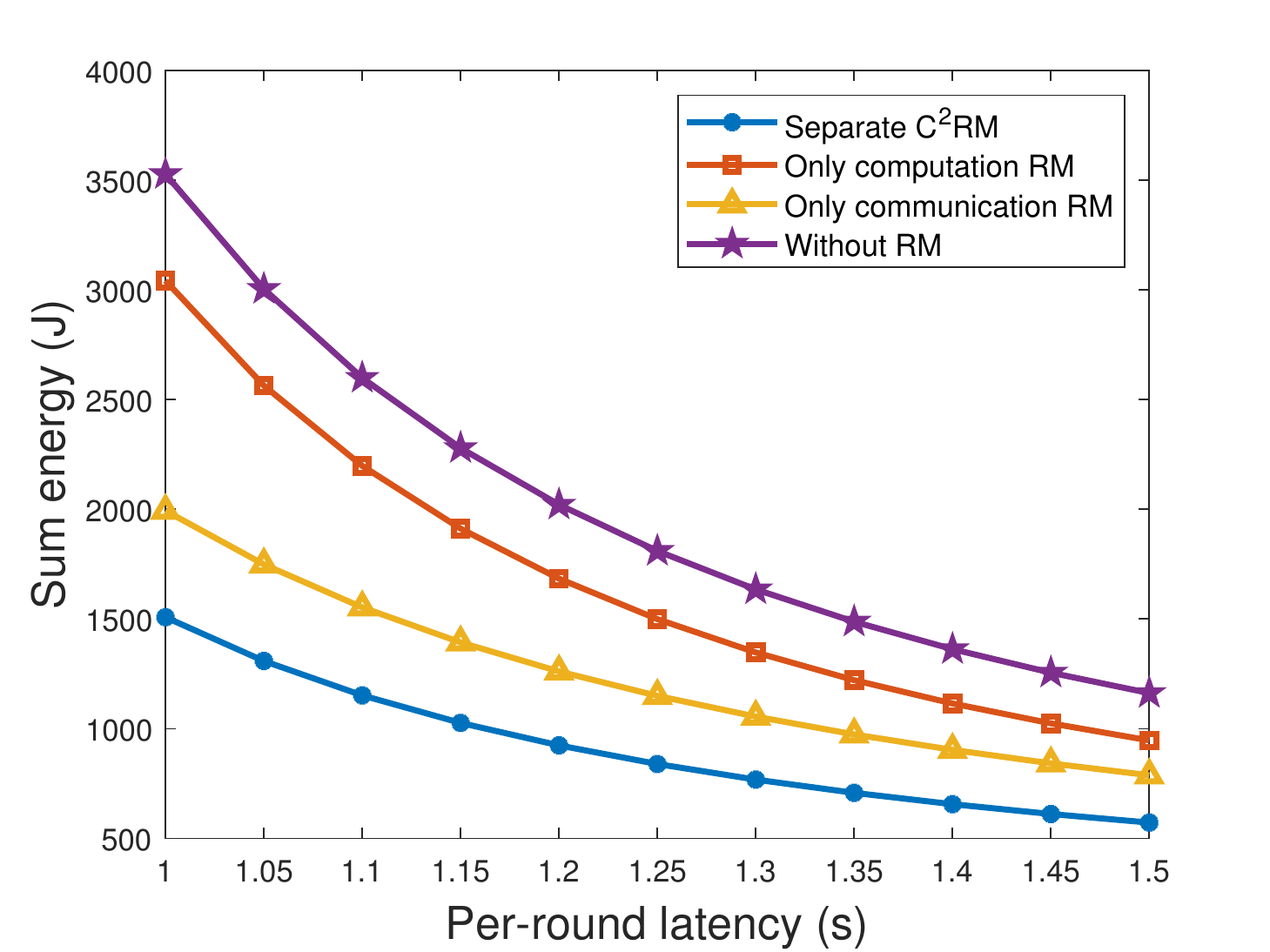}}
    \hspace{-2mm}
    \subfigure[Optimal time division]{
    \label{Fig: optimal time division}
    \includegraphics[width=8.0cm]{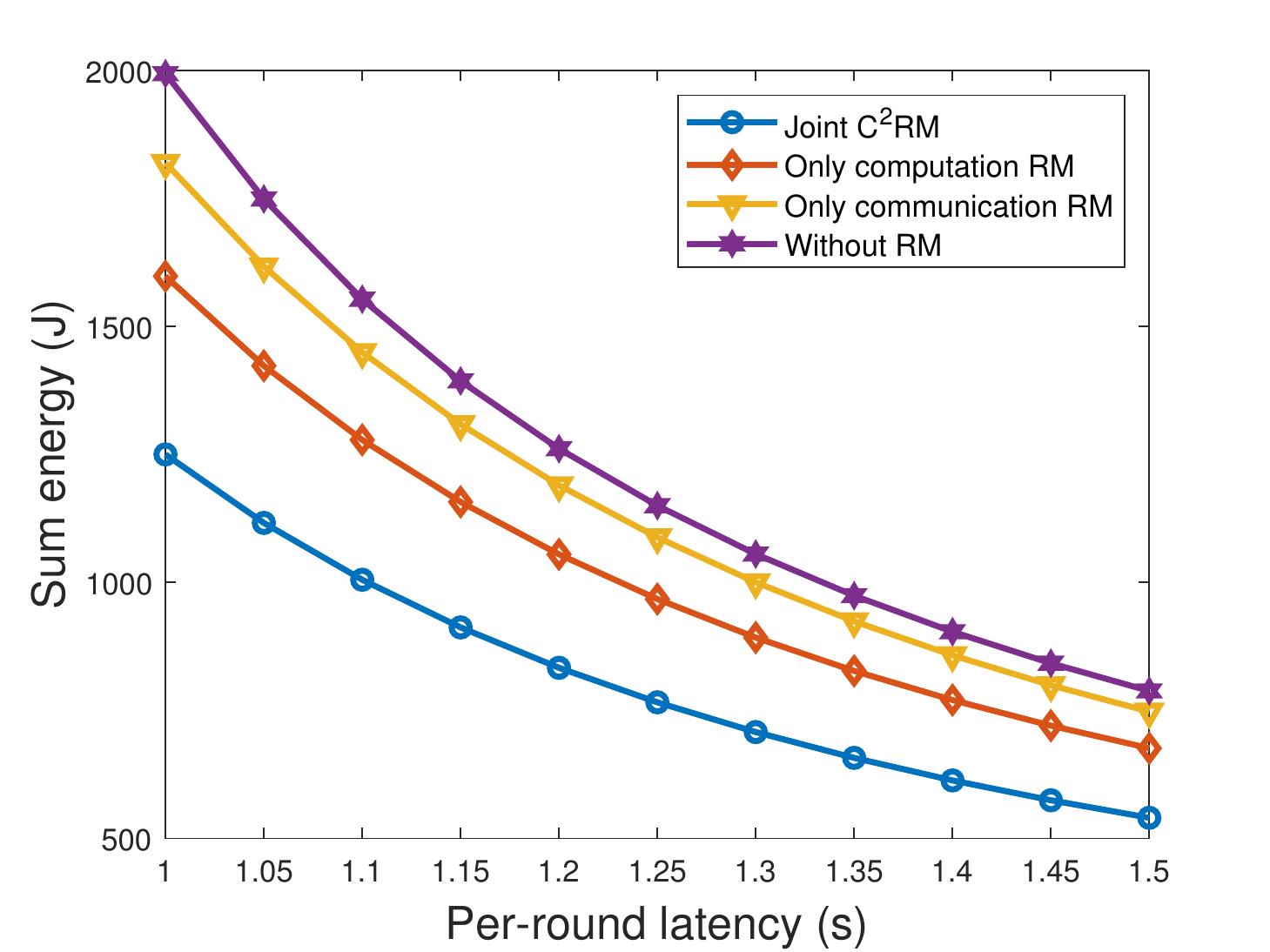}}
    \caption{Comparison of the energy efficiencies of $\text{C}^2$RM, only computation RM, only communication RM, and without RM schemes, with the (a) uniform time division and (b) optimal time divisions, respectively.}
    \label{Fig: C2RM}
\end{figure}

The performance of the proposed computation RM, communication RM and joint $\text{C}^2$RM policies are evaluated by simulations. To be specific, we consider two settings: a) uniform time division [see Fig. \ref{Fig: uniform time division}]; and b) optimal time divisions [see Fig. \ref{Fig: optimal time division}]. Under each setting, we consider four schemes: 1) $\text{C}^2$RM, 2) only computation RM, 3) only communication RM, and 4) without RM. In particular, to distinguish our proposed optimal $\text{C}^2$RM policy, we name the policy with $\text{C}^2$RM under uniform time division setting as ``separate $\text{C}^2$RM" while the proposed optimal one is called ``joint $\text{C}^2$RM". The curves of the sum energy versus the per-round latency $T$ are illustrated in Fig. \ref{Fig: C2RM}. Serveral observations can be made. First, the sum energy reduces as $T$ grows for all cases. This coincides with the results in Lemma \ref{lemma: energy-time rate} that the energy consumption is a monotonically decreasing function in computation and communication time. Second, for either setting a) or b), it can be found that the two schemes 2) and 3), namely only computation RM and only communication RM, outperform the scheme without RM but underperform the $\text{C}^2$RM. For example, under setting a), they reduce the sum energy of the policy without RM by $13.8\%$ and $43.5\%$, respectively, for per-round latency equal to $1.0$ s. Meanwhile, the separate $\text{C}^2$RM scheme reduces the sum energy of the policy without RM by $57.3\%$. These results demonstrate the effectiveness of our proposed workload and bandwidth allocation policies for $\text{C}^2$RM. Third, comparing the results in setting a) with b), we find that the strategy of $\text{C}^2$ time division plays a significant role in energy efficiency. For example, the four schemes 1)-4) in b) with optimal time divisions reduce the sum energy of those in a) by $17.2\%$, $47.5\%$, $9.8\%$, and $43.5\%$, respectively, for per-round latency equal to $1.0$ s. This coincides with Remark 4 that the optimal $\text{C}^2$ time division achieves the best energy efficiency by balancing $\text{C}^2$ heterogeneity. Fourth, our proposed joint $\text{C}^2$RM policy outperforms all other schemes, showing its effectiveness. For example, it reduces the sum energy of the schemes 1)-4) in a) by $17.2\%$, $59.0\%$, $37.4\%$, and $64.6\%$, as well as the schemes 2)-4) in b) by $21.9\%$, $31.3\%$, and $37.4\%$, respectively, for per-round latency equal to $1.0$ s.

\begin{figure}[t!]
\centering
\includegraphics[width=0.55\textwidth]{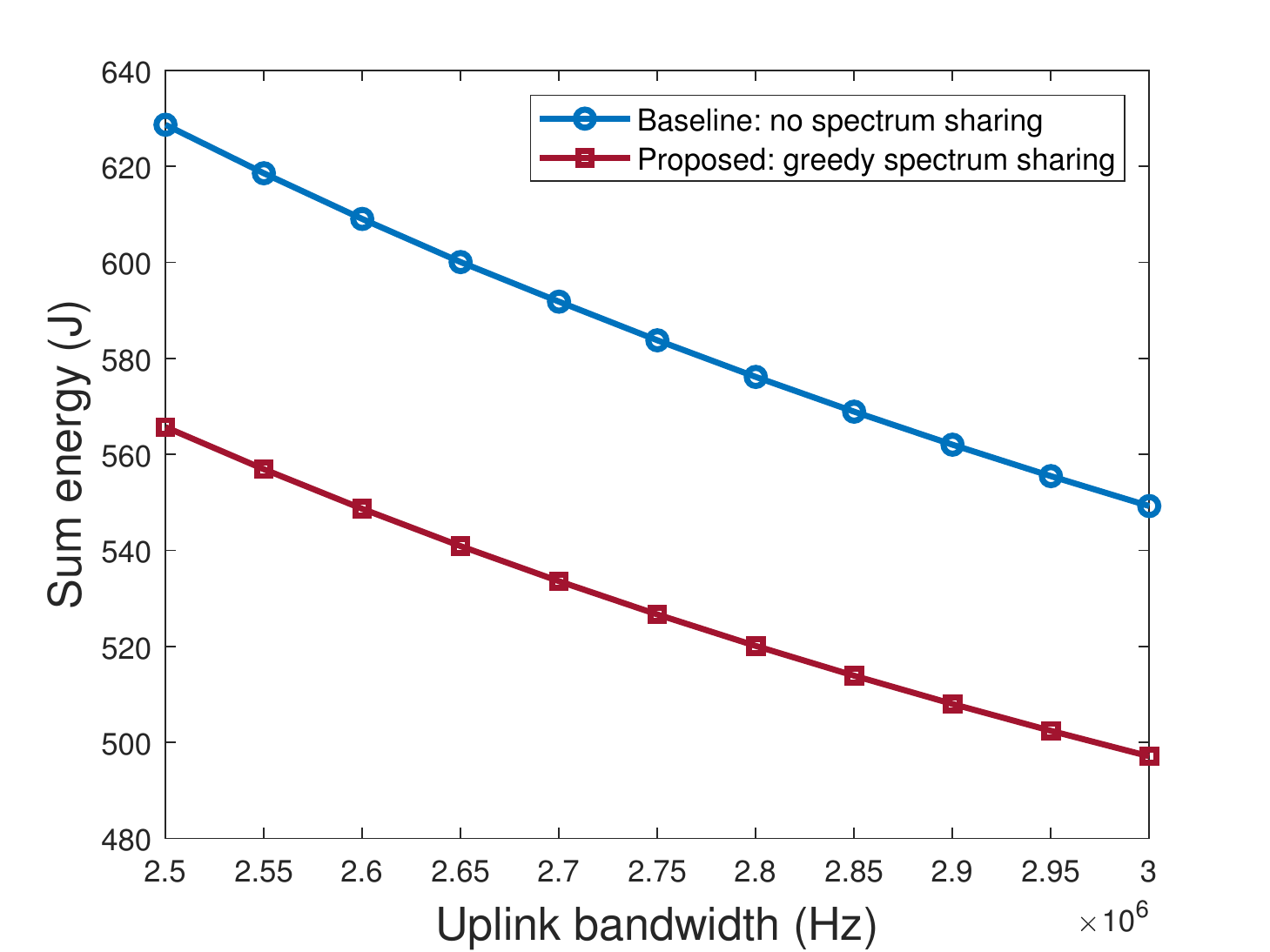}
\caption{Comparison between joint $\text{C}^2$RM with and without greedy spectrum sharing. The lines show sum energy vs. uplink bandwidth with fixed one round time $T=1$ s.}
\label{Fig: spectrum sharing}
\end{figure}

\subsubsection{Greedy spectrum sharing} 
The performance of the proposed greedy spectrum sharing algorithm in Algorithm~\ref{algorithm: greedy spectrum sharing} is benchmarked against the optimal $\text{C}^2$RM without spectrum sharing in Algorithm~\ref{Algorithm:joint_RM}. Note that the number of devices are set to be $K=20$ with high heterogeneity regarding their computation efficiencies and channels. The curves of sum energy versus the bandwidth $B$ are plotted in Fig.~\ref{Fig: spectrum sharing}. Two observations can be made. First, the sum energy reduces as $B$ grows as more bandwidths can be traded for lower transmission power. Second, it can be found that the proposed greedy spectrum sharing policy improves the energy efficiency of the $\text{C}^2$RM without spectrum sharing by scavenging unused radio resources. For example, it reduces sum energy of the baseline scheme, namely optimal $\text{C}^2$RM without spectrum sharing by $10.6\%$ for uplink bandwidth equal to $2.5$ MHz.

\subsection{Device Scheduling}

Consider the scenario that the edge server schedules a subset of devices for learning. The communication round  number is fixed as $10$ with $T=1$ s for each round. The performance of the proposed $\text{C}^2$-aware scheduling is benchmarked against the random selection scheme. The effects of the number of scheduled devices $M$ on the average learning accuracy of the FEEL algorithm and sum energy consumption are shown in Fig.~\ref{Fig:user_selection}.  
Several observations can be made as follows. First, the average learning accuracy is an increasing function of $M$ as it increases the training data per round. 
Second, it can be observed that the increase on $M$ leads to the growth of sum energy at a rate faster than linear, which agrees with the preceding discussion. 
Furthermore, one can observe that the proposed scheduling scheme outperforms the baseline which randomly selects $M$ devices, e.g., achieving $39.8\%$ sum energy reduction for $M=35$.

\begin{figure}[t!]
    \centering
    \includegraphics[width=0.55\textwidth]{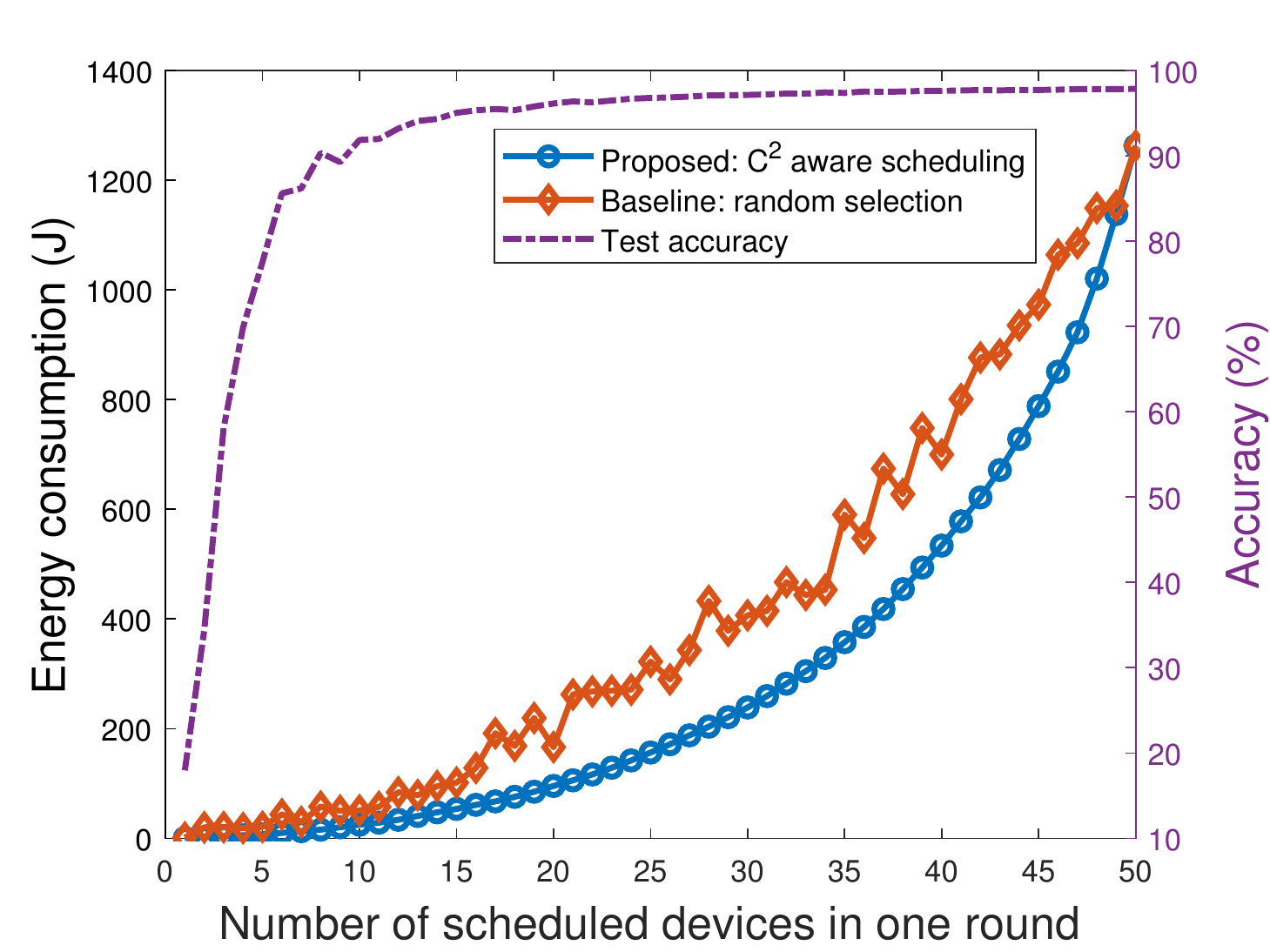}
    \caption{The comparison of the energy efficiencies between $\text{C}^2$ aware scheduling and random selection is shown by solid lines. The relationship between test accuracy and the number of scheduled devices per round is illustrated by the dash line.}
    \label{Fig:user_selection}
\end{figure}

\appendix

\subsection{Proof of Lemma \ref{lemma: computation condition}}\label{proof: computation condition}
Consider two cases for the second constraint in Problem {\bf P2} as follows. 

First, we consider the case that $\frac{W_k}{f_k}\geq \frac{W_k'}{f_k'}, k\in\mathcal{K}$. Then, we know ${t'}_k=\frac{W_k}{f_k}$. Therefore, we have $f_k=\frac{W_k}{{t'}_k}$ and $f_k'\geq \frac{W_k'}{{t'}_k}$.
It follows that the computation energy consumption is 
\begin{align}
    E_k^{\text{cmp}}=C_kW_kf_k^2+G_kW_k'f_k'^2\geq C_kW_k\frac{W_k^2}{{t'}_k^2}+G_kW_k'\frac{W_k'^2}{{t'}_k^2}.
\end{align}
The equality holds if and only if $f_k'=\frac{W_k'}{{t'}_k}$, meaning $\frac{W_k}{f_k}=\frac{W_k'}{f_k'}$. 

This is also true for the case $\frac{W_k}{f_k}\leq \frac{W_k'}{f_k'}$ due to its symmetric form.
\subsection{Proof of Lemma \ref{lemma: energy-load rate}}\label{proof: energy-load rate}
The computation energy of device $k$ can be expressed as
\begin{align}
    E_k^{\text{cmp}}=E_k^{\text{CPU}}(W_k,{t'}_k)+E_k^{\text{GPU}}(W_k',{t'}_k),\quad k\in\mathcal{K}.
\end{align}
The partial Lagrangian function is defined as
\begin{equation}
\begin{aligned}
    &\mathcal{L}(\{W_k\},\{W'_k\},\{\gamma_k\},\{\lambda_k\},\{\theta_k\})\\
    =&\sum_{k=1}^K\left[ E_k^{\text{CPU}}(W_k,{t'}_k)+E_k^{\text{GPU}}(W_k',{t'}_k)+\gamma_k(W_k+{W'}_k-W)-\lambda_kW_k-\theta_k{W'}_k\right],
\end{aligned}
\end{equation}
where $\{\lambda_k\geq 0\}$ and $\{\theta_k\geq 0\}$ are the Lagrange multipliers. Then we have the conditions:
\begin{align}
    \frac{\partial \mathcal{L}}{\partial W_k}=\frac{\partial E_k^{\text{CPU}}}{\partial W_k}+\gamma_k-\lambda_k=0,\quad  \frac{\partial \mathcal{L}}{\partial {W'}_k}=\frac{\partial E_k^{\text{GPU}}}{\partial {W'}_k}+\gamma_k-\theta_k=0,\quad\forall k\in\mathcal{K}.
\end{align}
The complementary slackness conditions give that
\begin{equation}\label{eqn: muW=0 and nuW=0}
    \lambda_kW_k=0,\qquad \theta_k{W'}_k=0,\qquad \forall k\in\mathcal{K}.
\end{equation}
Since at least one processor should be active. 
Assume ${W'}_k>0$ and then we have $\theta_k=0$ from the second condition in \eqref{eqn: muW=0 and nuW=0}. Since the multiplier $\lambda_k$ is required to be $\lambda_k\geq 0$, we know that
\begin{align}\label{RR46}
    \lambda_k=\frac{\partial E_k^{\text{CPU}}}{\partial W_k}-\frac{\partial E_k^{\text{GPU}}}{\partial {W'}_k}\geq 0,\quad\forall k\in\mathcal{K}.
\end{align}
Substituting results \eqref{Def: energy-load rate} of Lemma \ref{lemma: energy-load rate} into (\ref{RR46}), we have 
$W_k\geq\sqrt{\frac{G_k}{C_k}}{W'}_k > 0,~\forall k\in\mathcal{K}.$
Accordingly, we have $W_k>0$ so that $\lambda_k=0$. Therefore, it is optimal for both processors to be active with the energy-workload rate equilibrium \eqref{Equilibrium: energy-load rate} as stated in Lemma \ref{lemma: energy-load rate}.


\subsection{Proof of Lemma \ref{lemma: energy-bandwidth rate}}\label{proof: energy-bandwidth rate}
Substituting $t_k^* = T_k$ into \eqref{Eq:E_upload}, the communication energy of device $k$ can thus be expressed as the function of allocated bandwidths, i.e., $E_k^{\text{cmm}}(B_k)$.
By introducing Lagrange multipliers $\{\mu_k^*\}$ for the inequality constraints $\{B_k\geq 0\}$ and a scalar multiplier $\nu^*$ for the equality constraint $\sum_{k=1}^KB_k=B$, the Lagrangian function is defined as
\begin{equation}
    \mathcal{L}(\{B_k\},\{\mu_k\},\nu)=\sum_{k=1}^K\left[ E_k^{\text{cmm}}(B_k)+\mu_kB_k\right]+\nu\left(\sum_{k=1}^KB_k-B\right),
\end{equation}
Then we have the following conditions:
\begin{align}\label{eqns:KKT}
    B_k^*\geq 0,\quad \sum_{k=1}^KB_k^*=B,\quad \mu_k^*\geq 0,\quad \mu_k^*B_k^*=0, \quad \frac{\partial E_k^{\text{cmm}}}{\partial B_k}\bigg|_{B_k^*}-\mu_k^*+\nu^* = 0, \quad \forall k\in\mathcal{K}.
\end{align}
The feasible requirement gives $B_k^*>0,\forall k\!\in\!\mathcal{K}$, so that $\mu_k^*\!=\!0$. Therefore, $\frac{\partial E_k^{\text{cmm}}}{\partial B_k}\Big|_{B_k^*}\!\!\!=\!-\nu^*,~\forall k\!\in\!\mathcal{K}.$

\subsection{Proof of Lemma~\ref{Pre:updating_rule}}\label{proof: algorithm 1}
To begin with, we introduce several basic properties: 1) $B_k$ is a decreasing convex function of $\nu$ with the form in Lemma \ref{lemma: communication bandwidth allocation}; 2) $\nu_k$ is a decreasing and strictly convex function of $B_k$ with the form in \eqref{eqn: define nu_k}; 3) $\nu_k(B_k)>0$ for all $k$ and all value of $B_k>0$.
The proofs are straightforward and omitted for brevity.

Assume the point $\nu^{(i-1)}>\nu^*$, it follows from property 1) that $B_k(\nu^{(i-1)})<B_k^*$, $\forall k\in\mathcal{K}$. Denote $B_k^{(i)}\triangleq B_k(\nu^{(i-1)})$, and it holds that $\sum_{k=1}^KB_k^{(i)}<\sum_{k=1}^KB_k^*=B$. Therefore, $\text{sgn}\left(\sum_{k=1}^KB_k^{(i)}-B\right)=-1$ and
  $ \widetilde{B}_k^{(i)}=\frac{B}{\sum_{k=1}^KB_k^{(i)}}B_k^{(i)}>B_k^{(i)}.$
By property 2), we know that $\nu_k(\widetilde{B}_k^{(i)})<\nu_k(B_k^{(i)})=\nu^{(i-1)}$, $\forall k\in\mathcal{K}$. Then, we have $\max\limits_{k}\{\nu_k(\widetilde{B}_k^{(i)})\}<\nu^{(i-1)}$. 

Next, we prove by contradiction that $\max\limits_{k}\{\nu_k(\widetilde{B}_k^{(i)})\}>\nu^*$ as follows. 
First, we assume that $\max\limits_{k}\{\nu_k(\widetilde{B}_k^{(i)})\}<\nu^*$ and thus one can have $\nu_k(\widetilde{B}_k^{(i)})<\nu^*,~\forall k\in\mathcal{K}$. Accordingly, we have $\tilde{B}_k^{(i)}>B_k^*,~\forall k\in\mathcal{K}$, which implies that $\sum_{k=1}^K\widetilde{B}_k^{(i)}>\sum_{k=1}^KB_k^*=B$. However, it is invalid as
   $ \sum_{k=1}^K\widetilde{B}_k^{(i)}=\sum_{k=1}^K\frac{B}{\sum_{k=1}^KB_k^{(i)}}B_k^{(i)}=\frac{B}{\sum_{k=1}^KB_k^{(i)}}\sum_{k=1}^KB_k^{(i)}=B.$
Thus the earlier assumption is false. Thereby, we have $\max\limits_{k}\{\nu_k(\widetilde{B}_k^{(i)})\}>\nu^*$. 

Combining the results above, we have $\text{sgn}\left(\sum_{k=1}^KB_k^{(i)}-B\right)=-1$ and $\nu^*<\max\limits_{k}\{\nu_k(\widetilde{B}_k^{(i)})\}<\nu^{(i-1)}$. Following the same procedure, we can derive for the case of $\nu^{(i-1)}<\nu^*$, that $y = \text{sgn}\left(\sum_{k=1}^KB_k^{(i)}-B\right)=1$ and $\nu^{(i-1)}<\min\limits_{k}\{\nu_k(\widetilde{B}_k^{(i)})\}<\nu^*$. 

\subsection{Proof of Theorem \ref{theorem: convergence rate}}\label{proof: convergence rate}
At the beginning of the ($i+1$)-th round, all the devices receive the global model ${\mathbf{w}}^{(i)}$. However, the edge server only schedules $M$ devices for gradient computation, which makes the problem sophisticated to solve. To tackle the challenge, we use a trick that we assume all the devices compute the gradients $\{\nabla F_k({\mathbf{w}}^{(i)})\}$ based on their local datasets while only $M$ of them are aggregated for global model updating. This is equivalent to the learning process in our scenario. Then, the global model update rule can be written as
\begin{equation}
    {\mathbf{w}}^{(i+1)}={\mathbf{w}}^{(i)}-\frac{\eta}{M}\sum_{k\in\mathcal{M}_i}\nabla F_k({\mathbf{w}}^{(i)})=\frac{1}{M}\sum_{k\in\mathcal{M}_i}\left({\mathbf{w}}^{(i)}-\eta\nabla F_k({\mathbf{w}}^{(i)})\right).
\end{equation}
We define the virtual local updated model at device $k$ as $\mathbf{w}_k^{(i+1)}$ and denote it as
\begin{align}
    \mathbf{w}_k^{(i+1)}={\mathbf{w}}^{(i)}-\eta\nabla F_k({\mathbf{w}}^{(i)}).
\end{align}
As a result, the global update rule is equivalent to 
\begin{equation}
    {\mathbf{w}}^{(i+1)}=\frac{1}{M}\sum_{k\in\mathcal{M}_i}\mathbf{w}_k^{(i+1)}.
\end{equation}
Next, we introduce another virtual sequence as the average model over all virtual local updates
\begin{equation}
    \bar{\mathbf{w}}^{(i+1)}=\frac{1}{K}\sum_{k=1}^K\mathbf{w}_k^{(i+1)}.
\end{equation}
Accordingly, we have
\begin{equation}\label{RR53}
    \bar{\mathbf{w}}^{(i+1)}={\mathbf{w}}^{(i)}-\frac{\eta}{K}\sum_{k=1}^K\nabla F_k({\mathbf{w}}^{(i)}).
\end{equation}
The averaged virtual model shifts at the end of one round is
\begin{equation}
\begin{aligned}
    &\frac{1}{K}\sum_{k=1}^K\left\|\mathbf{w}_k^{(i+1)}-\bar{\mathbf{w}}^{(i+1)}\right\|^2=\frac{\eta^2}{K}\sum_{k=1}^K\left\|\nabla F_k({\mathbf{w}}^{(i)})-\frac{1}{K}\sum_{k=1}^K\nabla F_k({\mathbf{w}}^{(i)})\right\|^2\\
    &=\frac{\eta^2}{K}\sum_{k=1}^K\left\|\nabla F_k({\mathbf{w}}^{(i)})-\nabla F({\mathbf{w}}^{(i)})\right\|^2=\eta^2\mathbb{E}\left[\left\|\nabla F_k({\mathbf{w}}^{(i)})-\nabla F({\mathbf{w}}^{(i)})\right\|^2\right]\leq \eta^2\sigma^2.
\end{aligned}
\end{equation}
We further notice that $\|{\mathbf{w}}^{(i)}-\mathbf{w}^*\|^2=\|{\mathbf{w}}^{(i)}-\bar{\mathbf{w}}^{(i)}+\bar{\mathbf{w}}^{(i)}-\mathbf{w}^*\|^2$, and thus we have
\begin{align}
    &\|{\mathbf{w}}^{(i)}-\mathbf{w}^*\|^2=\|{\mathbf{w}}^{(i)}-\bar{\mathbf{w}}^{(i)}\|^2+\|\bar{\mathbf{w}}^{(i)}-\mathbf{w}^*\|^2+2\langle{\mathbf{w}}^{(i)}-\bar{\mathbf{w}}^{(i)},\bar{\mathbf{w}}^{(i)}-\mathbf{w}^*\rangle
\end{align}
Due to the fact 
    $\mathbb{E}_{\mathcal{M}_i}\langle{\mathbf{w}}^{(i)}-\bar{\mathbf{w}}^{(i)},\bar{\mathbf{w}}^{(i)}-\mathbf{w}^*\rangle=0$,
we have
\begin{align}
    \mathbb{E}_{\mathcal{M}_i}\|{\mathbf{w}}^{(i)}-\mathbf{w}^*\|^2=\mathbb{E}_{\mathcal{M}_i}\|{\mathbf{w}}^{(i)}-\bar{\mathbf{w}}^{(i)}\|^2+\mathbb{E}_{\mathcal{M}_i}\|\bar{\mathbf{w}}^{(i)}-\mathbf{w}^*\|^2.\label{r1:3}
\end{align}
The first term in (\ref{r1:3}) can be decomposed into the following two terms.
\begin{equation}\label{r2:1}
\begin{aligned}
    &\|{\mathbf{w}}^{(i)}-\bar{\mathbf{w}}^{(i)}\|^2=\left\|\frac{1}{M}\sum_{k\in\mathcal{M}_i}\mathbf{w}_k^{(i)}-\bar{\mathbf{w}}^{(i)}\right\|^2=\frac{1}{M^2}\left\|\sum_{k\in\mathcal{M}_i}\left(\mathbf{w}_k^{(i)}-\bar{\mathbf{w}}^{(i)}\right)\right\|^2\\
    &=\frac{1}{M^2}\left(\sum_{k\in\mathcal{M}_i}\|\mathbf{w}_k^{(i)}-\bar{\mathbf{w}}^{(i)}\|^2+\sum_{\substack{k,l\in\mathcal{M}_i\\k\neq l}}\langle \mathbf{w}_k^{(i)}-\bar{\mathbf{w}}^{(i)},\mathbf{w}_l^{(i)}-\bar{\mathbf{w}}^{(i)}\rangle\right).
\end{aligned}
\end{equation}
Taking expectation on the first term in (\ref{r2:1}), we have
\begin{equation}\label{r2:2}
\begin{aligned}
    &\mathbb{E}_{\mathcal{M}_i}\left[\sum_{k\in\mathcal{M}_i}\|\mathbf{w}_k^{(i)}-\bar{\mathbf{w}}^{(i)}\|^2\right]=\sum_{\substack{\mathcal{M}\subseteq\mathcal{K}\\|\mathcal{M}|=M}}{\rm Pr}(\mathcal{M}_i=\mathcal{M})\sum_{k\in\mathcal{M}_i}\|\mathbf{w}_k^{(i)}-\bar{\mathbf{w}}^{(i)}\|^2\\
    &=\sum_{\substack{\mathcal{M}\subseteq\mathcal{K}\\|\mathcal{M}|=M}}{\rm Pr}(\mathcal{M}_i=\mathcal{M})\sum_{k\in\mathcal{K}}{\rm Pr}(k\in\mathcal{M}_i)\|\mathbf{w}_k^{(i)}-\bar{\mathbf{w}}^{(i)}\|^2\\
    &=\frac{\binom{K-1}{M-1}}{\binom{K}{M}}\sum_{k\in\mathcal{K}}\|\mathbf{w}_k^{(i)}-\bar{\mathbf{w}}^{(i)}\|^2=\frac{M}{K}\sum_{k=1}^K\|\mathbf{w}_k^{(i)}-\bar{\mathbf{w}}^{(i)}\|^2.
\end{aligned}
\end{equation}
Taking expectation on the second term in (\ref{r2:1}), we have
\begin{equation}\label{r2:3}
\begin{aligned}
    &\mathbb{E}_{\mathcal{M}_i}\left[\sum_{\substack{k,l\in\mathcal{M}_i\\k\neq l}}\langle \mathbf{w}_k^{(i)}-\bar{\mathbf{w}}^{(i)},\mathbf{w}_l^{(i)}-\bar{\mathbf{w}}^{(i)}\rangle\right]=\sum_{\substack{\mathcal{M}\subseteq\mathcal{K}\\|\mathcal{M}|=M}}{\rm Pr}(\mathcal{M}_i=\mathcal{M})\sum_{\substack{k,l\in\mathcal{M}_i\\k\neq l}}\langle \mathbf{w}_k^{(i)}-\bar{\mathbf{w}}^{(i)},\mathbf{w}_l^{(i)}-\bar{\mathbf{w}}^{(i)}\rangle\\
    &=\frac{M(M-1)}{K(K-1)}\left(\left\|\sum_{k=1}^K(\mathbf{w}_k^{(i)}-\bar{\mathbf{w}}^{(i)})\right\|^2-\sum_{k=1}^K\left\|\mathbf{w}_k^{(i)}-\bar{\mathbf{w}}^{(i)}\right\|^2\right)=-\frac{M(M-1)}{K(K-1)}\sum_{k=1}^K\left\|\mathbf{w}_k^{(i)}-\bar{\mathbf{w}}^{(i)}\right\|^2.
\end{aligned}
\end{equation}
Combining (\ref{r2:2}) and (\ref{r2:3}), we have the result as follow:
\begin{equation}\label{r7}
    \mathbb{E}_{\mathcal{M}_i}\|{\mathbf{w}}^{(i)}-\bar{\mathbf{w}}^{(i)}\|^2=\frac{K-M}{MK(K-1)}\sum_{k=1}^K\left\|\mathbf{w}_k^{(i)}-\bar{\mathbf{w}}^{(i)}\right\|^2\leq \frac{K-M}{M(K-1)}\times \eta^2\sigma^2.
\end{equation}
Then, we consider the second term in (\ref{r1:3}). According to the update rule (\ref{RR53}), we know
\begin{equation}\label{r11}
\begin{aligned}
    &\|\bar{\mathbf{w}}^{(i+1)}-\mathbf{w}^*\|^2
    =\left\|{\mathbf{w}}^{(i)}-\frac{\eta}{K}\sum_{k=1}^K\nabla F_k({\mathbf{w}}^{(i)})-\mathbf{w}^*\right\|^2\\
    &=\|{\mathbf{w}}^{(i)}-\mathbf{w}^*\|^2+\eta^2\left\|\frac{1}{K}\sum_{k=1}^K\nabla F_k({\mathbf{w}}^{(i)})\right\|^2-\frac{2\eta}{K}\sum_{k=1}^K\langle \nabla F_k({\mathbf{w}}^{(i)}), {\mathbf{w}}^{(i)}-\mathbf{w}^*\rangle.
\end{aligned}
\end{equation}
According to the smoothness property in Assumption \ref{assumption: convexity and smoothness}, the second term in (\ref{r11}) is bounded with
\begin{align}\label{r12}
    \left\|\frac{1}{K}\sum_{k=1}^K\nabla F_k({\mathbf{w}}^{(i)})\right\|^2=\left\|\nabla F({\mathbf{w}}^{(i)})\right\|^2\leq 2\beta\left(F({\mathbf{w}}^{(i)})-F(\mathbf{w}^*)\right).
\end{align}
According to the convexity property in Assumption \ref{assumption: convexity and smoothness}, we know that
\begin{align}
    \langle \nabla F_k({\mathbf{w}}^{(i)}),\mathbf{w}^*- {\mathbf{w}}^{(i)}\rangle\leq F_k(\mathbf{w}^*)-F_k({\mathbf{w}}^{(i)}).
\end{align}
Therefore, the third term in (\ref{r11}) is bounded with
\begin{equation}\label{r15}
-\frac{2\eta}{K}\sum_{k=1}^K\langle \nabla F_k({\mathbf{w}}^{(i)}), {\mathbf{w}}^{(i)}-\mathbf{w}^*\rangle\leq\frac{2\eta}{K}\sum_{k=1}^K\left(F_k(\mathbf{w}^*)-F_k({\mathbf{w}}^{(i)})\right)=2\eta\left(F(\mathbf{w}^*)-F({\mathbf{w}}^{(i)})\right).
\end{equation}
Combining the results in (\ref{r11}), (\ref{r12}) and (\ref{r15}), we obtain that
\begin{equation}\label{r19}
\begin{aligned}
    &\|\bar{\mathbf{w}}^{(i+1)}-\mathbf{w}^*\|^2=\|{\mathbf{w}}^{(i)}-\mathbf{w}^*\|^2+\eta^2\left\|\frac{1}{K}\sum_{k=1}^K\nabla F_k({\mathbf{w}}^{(i)})\right\|^2-\frac{2\eta}{K}\sum_{k=1}^K\langle \nabla F_k({\mathbf{w}}^{(i)}), {\mathbf{w}}^{(i)}-\mathbf{w}^*\rangle\\
    &\leq \|{\mathbf{w}}^{(i)}-\mathbf{w}^*\|^2+2\eta^2\beta\left(F({\mathbf{w}}^{(i)})-F(\mathbf{w}^*)\right)+2\eta\left(F(\mathbf{w}^*)-F({\mathbf{w}}^{(i)})\right)\\
    &=\|{\mathbf{w}}^{(i)}-\mathbf{w}^*\|^2-2\eta(1-\eta\beta)\left(F({\mathbf{w}}^{(i)})-F(\mathbf{w}^*)\right).
\end{aligned}
\end{equation}
Assume $\eta<\frac{1}{\beta}$, so that $1-\eta\beta>0$. In addition, due to $\mathbf{w}^*=\arg\min_{\mathbf{w}}F(\mathbf{w})$, it is obvious that $F({\mathbf{w}}^{(i)})- F(\mathbf{w}^*)\geq 0$. Combining all the results above, we obtain the upper bound for (\ref{r1:3}) as
\begin{equation}\label{Rr64}
\begin{aligned}
    &\mathbb{E}_{\mathcal{M}_{i}}\left[\|{\mathbf{w}}^{(i+1)}-\mathbf{w}^*\|^2\right]=\mathbb{E}_{\mathcal{M}_{i}}\left[\|{\mathbf{w}}^{(i+1)}-\bar{\mathbf{w}}^{(i+1)}\|^2\right]+\mathbb{E}_{\mathcal{M}_{i}}\left[\|\bar{\mathbf{w}}^{(i+1)}-\mathbf{w}^*\|^2\right]\\
    &\leq \frac{\eta^2\sigma^2(K-M)}{M(K-1)}+\mathbb{E}_{\mathcal{M}_i}\left[\|{\mathbf{w}}^{(i)}-\mathbf{w}^*\|^2\right]-2\eta(1-\eta\beta)\mathbb{E}_{\mathcal{M}_i}\left[F({\mathbf{w}}^{(i)})-F(\mathbf{w}^*)\right].
\end{aligned}
\end{equation}
To ease the notation, denote that
\begin{align}
    a_{i}=\mathbb{E}_{\mathcal{M}_i}\left[\|{\mathbf{w}}^{(i)}-\mathbf{w}^*\|^2\right],\quad 
    b=\frac{\sigma^2(K-M)}{M(K-1)},\quad
    e_i=\mathbb{E}_{\mathcal{M}_i}\left[F({\mathbf{w}}^{(i)})-F(\mathbf{w}^*)\right]
\end{align}
Then \eqref{Rr64} becomes 
\begin{align}
    a_{i+1}\leq a_i+\eta^2b-2\eta(1-\eta\beta) e_i\label{r23}
\end{align}
Re-write (\ref{r23}) and arrange it in the following form:
\begin{gather}
    e_i\leq \frac{a_{i}- a_{i+1}+\eta^2b}{2\eta(1-\eta\beta)}.
\end{gather}
Thereby, the result in Theorem \ref{theorem: convergence rate} can be derived after taking expectation over rounds:
\begin{equation}
    \frac{1}{N}\sum_{i=0}^{N-1}e_i
    \leq\frac{\sum_{i=0}^{N-1} (a_{i}- a_{i+1})}{2\eta(1-\eta\beta)N}+\frac{\eta b}{2(1-\eta\beta)}\leq\frac{a_{0}}{2\eta(1-\eta\beta)N}+\frac{\eta b}{2(1-\eta\beta)}.
\end{equation}
If we further set the learning rate as $\eta=\frac{1}{\sqrt{N}}\leq\frac{1}{2\beta}$, we have
\begin{equation}
    \frac{1}{N}\sum_{i=0}^{N-1}e_i
    \leq\frac{a_{0}}{\eta N}+\eta b=\frac{1}{\sqrt{N}}\left[\|\mathbf{w}^{(0)}-\mathbf{w}^*\|^2+\frac{\sigma^2(K-M)}{(K-1)M}\right].
\end{equation}

\subsection{Clarification of Definition \ref{definition: acceleration}}\label{proof: greedy spectrum sharing}
The energy consumption of device $k$ can be expressed as
\begin{equation}
    E_k=E_k^{\text{cmp}}(W_k,W_k',{t'}_k)+E_k^{\text{cmm}}(B_k,t_k)
\end{equation}
with $W_k+W_k'=W$ and $t_k'+t_k=T$.

Since $W_k$ and $W'_k$ are irrelevant to the allocated bandwidth $B_k$, the derivative of energy with respect to bandwidth can be calculated as follows:
\begin{equation}
\begin{aligned}
    \frac{dE_k}{dB_k}&=\frac{\partial E_k^{\text{cmm}}}{\partial B_k}+\frac{\partial E_k^{\text{cmp}}}{\partial {t'}_k}\frac{d{t'}_k}{dB_k}+\frac{\partial E_k^{\text{cmm}}}{\partial t_k}\frac{dt_k}{dB_k}\\
    &=\frac{\partial E_k^{\text{cmm}}}{\partial B_k}+\xi_k\left(\frac{d{t'}_k}{dB_k}+\frac{dt_k}{dB_k}\right)=\frac{\partial E_k^{\text{cmm}}}{\partial B_k},
\end{aligned}
\end{equation}
where the energy-time rate equilibrium (see Lemma 5) gives $\frac{\partial E_k^{\text{cmp}}}{\partial {t'}_k}=\frac{\partial E_k^{\text{cmm}}}{\partial t_k}\triangleq \xi_k$, and the time constraint, namely ${t'}_k+t_k=T$, gives $\frac{d{t'}_k}{dB_k}+\frac{dt_k}{dB_k}=0$.

Then, we can derive the second derivative of energy with respect to bandwidth as follows:
\begin{equation}
\begin{aligned}
\frac{d^2E_k}{dB_k^2}&=\frac{d}{dB_k}\frac{\partial E_k^{\text{cmm}}}{\partial B_k}=\frac{\partial^2E_k^{\text{cmm}}}{\partial B_k^2}+\frac{\partial^2E_k^{\text{cmp}}}{\partial B_k\partial t'_k}\frac{dt'_k}{dB_k}+\frac{\partial^2E_k^{\text{cmm}}}{\partial B_k\partial t_k}\frac{dt_k}{dB_k}\\
&=\frac{\partial^2E_k^{\text{cmm}}}{\partial B_k^2}+\frac{\partial}{\partial B_k}\left(\frac{\partial E_k^{\text{cmp}}}{\partial t'_k}-\frac{\partial E_k^{\text{cmm}}}{\partial t_k}\right)\frac{dt'_k}{dB_k}=\frac{\partial^2E_k^{\text{cmm}}}{\partial B_k^2}.
\end{aligned}
\end{equation}


\vspace{2mm}

\bibliography{Energy_Efficient.bib}

\begin{thebibliography}{10}
\providecommand{\url}[1]{#1}
\csname url@samestyle\endcsname
\providecommand{\newblock}{\relax}
\providecommand{\bibinfo}[2]{#2}
\providecommand{\BIBentrySTDinterwordspacing}{\spaceskip=0pt\relax}
\providecommand{\BIBentryALTinterwordstretchfactor}{4}
\providecommand{\BIBentryALTinterwordspacing}{\spaceskip=\fontdimen2\font plus
\BIBentryALTinterwordstretchfactor\fontdimen3\font minus
  \fontdimen4\font\relax}
\providecommand{\BIBforeignlanguage}[2]{{%
\expandafter\ifx\csname l@#1\endcsname\relax
\typeout{** WARNING: IEEEtran.bst: No hyphenation pattern has been}%
\typeout{** loaded for the language `#1'. Using the pattern for}%
\typeout{** the default language instead.}%
\else
\language=\csname l@#1\endcsname
\fi
#2}}
\providecommand{\BIBdecl}{\relax}
\BIBdecl

\bibitem{gxzhu_2018_edge_learning}
G.~{Zhu}, D.~{Liu}, Y.~{Du}, C.~{You}, J.~{Zhang}, and K.~{Huang}, ``Toward an
  intelligent edge: Wireless communication meets machine learning,'' \emph{IEEE
  Commun. Mag.}, vol.~58, no.~1, pp. 19--25, 2020.

\bibitem{gxzhu2018FEEL}
G.~{Zhu}, Y.~{Wang}, and K.~{Huang}, ``Broadband analog aggregation for
  low-latency federated edge learning,'' \emph{IEEE Trans. Wireless Commun.},
  vol.~19, no.~1, pp. 491--506, 2020.

\bibitem{deniz2019federated_edge_learning}
M.~M. {Amiri} and D.~{Gündüz}, ``Machine learning at the wireless edge:
  Distributed stochastic gradient descent over-the-air,'' \emph{IEEE Trans.
  Signal Process.}, vol.~68, pp. 2155--2169, 2020.

\bibitem{sqwang_2018_edge_learning}
S.~Wang, T.~Tuor, T.~Salonidis, K.~K. Leung, C.~Makaya, T.~He, and K.~Chan,
  ``When edge meets learning: Adaptive control for resource-constrained
  distributed machine learning,'' in \emph{Proc. IEEE Conf. Comput. Commun.
  ({INFOCOM})}, Honolulu, USA, Apr 16-19, 2018.

\bibitem{R9}
S.~Mittal and J.~S. Vetter, ``A survey of {CPU-GPU} heterogeneous computing
  techniques,'' \emph{ACM Comput. Surv.}, vol.~47, no.~4, 2015.

\bibitem{nishio2018background}
T.~{Nishio} and R.~{Yonetani}, ``Client selection for federated learning with
  heterogeneous resources in mobile edge,'' in \emph{Proc. IEEE Int. Conf.
  Commun. (ICC)}, Shanghai, China, May 20-24, 2019.

\bibitem{yang2018background}
K.~{Yang}, T.~{Jiang}, Y.~{Shi}, and Z.~{Ding}, ``Federated learning via
  over-the-air computation,'' \emph{IEEE Trans. Wireless Commun.}, vol.~19,
  no.~3, pp. 2022--2035, 2020.

\bibitem{liye}
C.~{Xiong}, G.~Y. {Li}, S.~{Zhang}, Y.~{Chen}, and S.~{Xu}, ``Energy-efficient
  resource allocation in {OFDMA} networks,'' \emph{IEEE Trans. Commun.},
  vol.~60, no.~12, pp. 3767--3778, 2012.

\bibitem{QS}
Q.~Zeng, Y.~Du, K.~Huang, and K.~K. Leung, ``Energy-efficient radio resource
  allocation for federated edge learning,'' in \emph{Proc. IEEE Int. Conf.
  Commun. (ICC) Workshop}, Dublin, Ireland, Jun 7-11, 2020.

\bibitem{chen2019joint}
M.~Chen, Z.~Yang, W.~Saad, C.~Yin, H.~V. Poor, and S.~Cui, ``A joint learning
  and communications framework for federated learning over wireless networks,''
  \emph{[Online] https://arxiv.org/pdf/1909.07972.pdf}, 2019.

\bibitem{sun2019energyaware}
Y.~Sun, S.~Zhou, and D.~Gündüz, ``Energy-aware analog aggregation for
  federated learning with redundant data,'' \emph{[Online]
  https://arxiv.org/pdf/1911.00188.pdf}, 2019.

\bibitem{tran2019energy}
C.~Dinh, N.~H. Tran, M.~N.~H. Nguyen, C.~S. Hong, W.~Bao, A.~Y. Zomaya, and
  V.~Gramoli, ``Federated learning over wireless networks: Convergence analysis
  and resource allocation,'' \emph{[Online]
  https://arxiv.org/pdf/1910.13067.pdf}, 2019.

\bibitem{yang2019energy}
Z.~Yang, M.~Chen, W.~Saad, C.~S. Hong, and M.~Shikh-Bahaei, ``Energy efficient
  federated learning over wireless communication networks,'' \emph{[Online]
  https://arxiv.org/pdf/1911.02417.pdf}, 2019.

\bibitem{mo2020energyefficient}
X.~Mo and J.~Xu, ``Energy-efficient federated edge learning with joint
  communication and computation design,'' \emph{[Online]
  https://arxiv.org/pdf/2003.00199.pdf}, 2020.

\bibitem{R15}
A.~K. {Singh}, K.~R. {Basireddy}, A.~{Prakash}, G.~V. {Merrett}, and B.~M.
  {Al-Hashimi}, ``Collaborative adaptation for energy-efficient heterogeneous
  mobile {SoCs},'' \emph{IEEE Trans. on Comput.}, vol.~69, no.~2, pp. 185--197,
  2020.

\bibitem{R16}
L.~W. Chang, J.~G. Luna, I.~E. Hajj, S.~Huang, D.~Chen, and W.~Hwu,
  ``Collaborative computing for heterogeneous integrated systems,'' in
  \emph{Proc. 8th ACM/SPEC on Int. Conf. on Perform. Eng.}, L'Aquila, Italy,
  Apr 22-26, 2017.

\bibitem{R8}
A.~Ignatov, R.~Timofte, W.~Chou, K.~Wang, M.~Wu, T.~Hartley, and L.~Van~Gool,
  ``{AI Benchmark: Running} deep neural networks on android smartphones,'' in
  \emph{Proc. Eur. Conf. on Comput. Vision (ECCV) Workshops}, Munich, Germany,
  Sep 8-14, 2018.

\bibitem{R18}
Y.~Kim, J.~Kim, D.~Chae, D.~Kim, and J.~Kim, ``$\mu${Layer:} low latency
  on-device inference using cooperative single-layer acceleration and
  processor-friendly quantization,'' in \emph{Proc. 14th Eur. Syst. Conf.},
  Dresden, Germany, Mar 25-28, 2019.

\bibitem{R11}
C.~{Chen}, K.~{Li}, A.~{Ouyang}, and K.~{Li}, ``Flinkcl: An opencl-based
  in-memory computing architecture on heterogeneous {CPU-GPU} clusters for big
  data,'' \emph{IEEE Trans. Comput.}, vol.~67, no.~12, pp. 1765--1779, 2018.

\bibitem{R14}
Y.~Lee, K.~G. Shin, and H.~S. Chwa, ``Thermal-aware scheduling for integrated
  {CPUs--GPU} platforms,'' \emph{ACM Trans. Embedded Comput. Syst.}, vol.~18,
  no.~5s, 2019.

\bibitem{R5}
S.~Naffziger, ``{AMD’s} commitment to accelerating energy efficiency,''
  \emph{[Online]https://www.amd.com/system/files/documents/\\energy-efficiency-whitepaper.pdf}.

\bibitem{R20}
J.~G. Park, C.~Y. Hsieh, N.~Dutt, and S.~S. Lim, ``Synergistic {CPU-GPU}
  frequency capping for energy-efficient mobile games,'' \emph{ACM Trans.
  Embedded Comput. Syst.}, vol.~17, no.~2, 2017.

\bibitem{R10}
S.~Rai and M.~Chaudhuri, ``Using criticality of {GPU} accesses in memory
  management for {CPU-GPU} heterogeneous multi-core processors,'' \emph{ACM
  Trans. Embedded Comput. Syst.}, vol.~16, no.~5s, 2017.

\bibitem{federatedlearning}
B.~McMahan, E.~Moore, D.~Ramage, S.~Hampson, and B.~A. y~Arcas,
  ``Communication-efficient learning of deep networks from decentralized
  data,'' in \emph{Proc. Int. Conf. Artif. Intell. and Statist.}, Fort
  Lauderdale, USA, Apr 20-22, 2017.

\bibitem{zhang2018flops}
X.~{Zhang}, X.~{Zhou}, M.~{Lin}, and J.~{Sun}, ``Shufflenet: An extremely
  efficient convolutional neural network for mobile devices,'' in \emph{Proc.
  IEEE/CVF Conf. Comput. Vision \& Pattern Recognit. (CVPR)}, Salt Lake City,
  USA, Jun 18-23, 2018.

\bibitem{partition}
T.~B.~Nun and T.~Hoefler, ``Demystifying parallel and distributed deep
  learning: An in-depth concurrency analysis,'' \emph{ACM Comput. Surv.},
  vol.~52, no.~4, 2019.

\bibitem{parallel}
S.~{Scanzio}, S.~{Cumani}, R.~{Gemello}, F.~{Mana}, and P.~{Laface}, ``Parallel
  implementation of artificial neural network training,'' in \emph{Proc. IEEE
  Int. Conf. on Acoust., Speech and Signal Process.}, Dallas, USA, Mar 15-19,
  2010.

\bibitem{liu2012dvfs}
{C. Liu}, {J. Li}, {W. Huang}, J.~{Rubio}, E.~{Speight}, and {F. Lin},
  ``Power-efficient time-sensitive mapping in heterogeneous systems,'' in
  \emph{Proc. 21st Int. Conf. Parallel Archit. and Compilation Tech. (PACT)},
  Minneapolis, USA, Sep 21-25, 2012.

\bibitem{signSGD}
J.~Bernstein, Y.~X. Wang, K.~Azizzadenesheli, and A.~Anandkumar, ``sign{SGD}:
  Compressed optimisation for non-convex problems,'' in \emph{Proc. 35th Int.
  Conf. Mach. Learn. (ICML)}, Stockholm, Sweden, Jul 11-13, 2018.

\bibitem{wen2020joint}
D.~Wen, M.~Bennis, and K.~Huang, ``Joint parameter-and-bandwidth allocation for
  improving the efficiency of partitioned edge learning,'' \emph{[Online]
  https://arxiv.org/pdf/2003.04544.pdf}, 2020.

\bibitem{changsheng}
C.~{You}, K.~{Huang}, H.~{Chae}, and B.~{Kim}, ``Energy-efficient resource
  allocation for mobile-edge computation offloading,'' \emph{IEEE Trans.
  Wireless Commun.}, vol.~16, no.~3, pp. 1397--1411, 2017.

\bibitem{khaled2019analysis}
A.~Khaled, K.~Mishchenko, and P.~Richtárik, ``First analysis of local {GD} on
  heterogeneous data,'' \emph{[Online] https://arxiv.org/pdf/1909.04715.pdf},
  2019.

\end{thebibliography}
\bibliographystyle{IEEEtran}

\end{document}